\documentclass[a4paper,11pt]{article}

\usepackage{amsthm,amsmath,amssymb,amsfonts,pstricks,pst-node,graphics}

\usepackage{amsmath,amsfonts,amsthm,relsize}
\usepackage{amsmath,amsfonts,amsthm,amssymb}
\usepackage{graphicx,pstricks}
\newtheorem{theorem}{Theorem}
\numberwithin{theorem}{section}

\newtheorem{example}[theorem]{Example}

\newtheorem{lemma}[theorem]{Lemma}

\newtheorem{proposition}[theorem]{Proposition}

\theoremstyle{definition}
\newtheorem{definition}[theorem]{Definition}

\def\th{\theta}
\def\et{\eta}
\def\cH{\mathcal{H}}
\def\bt{{{\bf {\xi}}}}
\def\la{\lambda}
\def\wdg{\wedge}

\def\RP{\mathbb {RP}}
\def\R{\mathbb R}
\def\Z{\mathbb Z}
\def\sl{\mathfrak{sl}}

\def\g{\mathfrak{g}}

\def\HH{\mathcal{H}}
\def\U{\mathcal{U}}
\def\G{\mathcal{G}}
\def\Ro{\mathcal{R}}
\def\K{\mathcal{K}}
\def\PSL{\mathrm{PSL}}
\def\SL{\mathrm{SL}}

\def\N{\mathcal{N}}
\def\F{\mathcal{F}}

\def\P{\mathcal{P}}
\def\Q{\mathcal{Q}}

\def\T{\mathcal{T}}

\def\tr{\mathrm{tr}}
\def\exp{\mathrm{exp}}

\def\kb{\mathbf{k}}

\def\f{{\bf f}}
\def\h{{\mathfrak h}}
\def\a{{\bf a}}
\def\b{{\bf b}}

\def\q{{\bf q}}
\def\m{{\mathfrak m}}
\def\v{{\bf v}}

\def\g{\mathfrak{g}}

\def\ha{f^s_a}
\def\hb{f^s_b}
\def\han{f_a}
\def\hbn{f_b}

\def\s{\varsigma}

\renewcommand{\k}[1]{\kappa_{#1}}
\begin{document}
\title{Hamiltonian evolutions of twisted gons in $\RP^n$}
\author{Gloria Mar\'i Beffa \& Jing Ping Wang\footnote{Supported in part by GMB's NSF grant DMS \#0804541 and JPW's EPSRC grant EP/I038659/1.}}
\date{}
\maketitle
\begin{abstract} In this paper we describe a well-chosen discrete moving frame and their associated invariants along projective polygons in $\RP^n$, and we use them to write explicit general expressions for invariant evolutions of projective $N$-gons. We then use a reduction process inspired by a discrete Drinfeld-Sokolov reduction to obtain a natural Hamiltonian structure on the space of projective invariants, and we establish a close relationship between the projective $N$-gon evolutions and the Hamiltonian evolutions on the invariants of the flow. We prove that {any} Hamiltonian evolution is induced on invariants by an evolution of $N$-gons - what we call a projective realization - and we give the direct connection.
Finally, in the planar case we provide completely integrable evolutions (the Boussinesq lattice related to the lattice $W_3$-algebra), their projective realizations and their Hamiltonian pencil.
We generalize both structures to $n$-dimensions and we prove that they are Poisson. 
We define explicitly the $n$-dimensional generalization of the planar evolution (the discretization of the $W_n$-algebra) and prove that it is completely integrable, providing also its projective realization.
\end{abstract}

\section{Introduction}
Studies have shown a close relationship between evolution of curves invariant under a group action and completely integrable systems. The best known such relation was established by Hasimoto in \cite{Ha} where he showed that the Vortex-Filament flow (VF) - a curve flow in Euclidean space, invariant under the Euclidean group - induces the nonlinear Shr\"odinger evolution (NLS) on the curvature and torsion of the curve flow. One can state this fact by describing VF as an Euclidean realization of NLS. Many such examples followed in \cite{M3, M4, M5, MSW, Q1, Q2, SW, TT, TU} and others, showing realizations in classical geometries, in the space of pure Spinors, Lagragian Grassmannians, and more. The systems realized include KdV, mKdV, Adler-Gel'fand-Dikii (AGD) flows (defined in fact by Lax), Sowada-Koterra systems, NLS, etc, with one equation sometimes realized in several different geometries.

All of these systems are biHamiltonian and at the heart of many of the studies one finds natural Hamiltonian structures defined on the curvatures of the flow. One such structure was obtained and proved to be linked directly to invariant curve evolutions in \cite{M1} and \cite{M2}, ensuring that {\it any} of its Hamiltonian evolutions had a geometric realization. The reduction of a second, compatible structure usually indicated the existence of an associated integrable system. Establishing the relation between Hamiltonian structures on curvatures and invariant curve flows was facilitated by the definition of group-based moving frames and the theoretical framework around this concept (see \cite{FO}). 

In a recent paper (\cite{MMW}) the authors developed the discrete version of group-based moving frames and used it to study induced completely integrable systems on discrete curvatures (or invariants) by invariant evolutions of polygons in different geometric settings. In particular they found projective and centro-affine discrete realizations of the modified Volterra and Toda lattices, and a realization of a Volterra-type equation in the Euclidean 2-sphere. They also obtained Hamiltonian structures although it was not at all clear how these could be obtained in general.

In this paper we study the case of twisted gons in $\RP^n$ (they are twisted to ensure that invariants are periodic). The space of gons, as related to integrable systems, has gained in relevance lately because of its connection to the pentagram map (see \cite{OST}). Indeed, the study of generalizations of the pentagram map was originally the motivation for GMB in the pursue of the discrete interpretation of a moving frame. In \cite{OST} the authors proved that the pentagram map (a very simple map taking a gon to the gon obtained by intersecting the segments that join every other vertex) was a completely integrable discretization of Boussinesq equation (or second AGD flow). Later work (\cite{Kh}) showed that some generalizations defined in \cite{M6} discretizing higher order AGD flows were also completely integrable. Also inspired by the pentagram map, the author of \cite{Ian} defined completely integrable systems on invariants of planar projective polygons, obtained by reduction from the centro-affine case. 
In this 
paper we do not work on maps but on differential-difference 
systems, but it is our hope that the resulting Poisson structures might also be relevant to this area. 

 In the first part of the paper we describe background information on discrete moving frames and Poisson-Lie groups. The paper is aimed at two different audiences, the math-physics/geometry audience and the computational/integrable systems one. Accordingly we have tried to include enough background for both and we have worked out the projective plane case in detail throughout the paper. In section 3 we proceed to describe a projective discrete moving frame along projective $N$-gons, choosing a particular frame that will produce invariants fitting our purposes. Section 4 describes how we can find explicitly, algorithmically and without any previous knowledge of the moving frame,  the evolution induced on the $N$-gon invariants by an invariant evolution of $N$-gons. The case of the projective plane already hints clearly at a direct connection to Hamiltonian systems at this stage. 
 
Section 5 proves that in the discrete case there also exists a naturally defined Hamiltonian structure on the space of $N$-gon invariants. We obtain this structure by linking discrete invariants to a reduction  process similar to that found in \cite{S2} and \cite{S3}, where the authors describe a discrete analogue of the Drinfel'd-Sokolov reduction (the DS reduction was directly used in \cite{Ian} to connect centro-affine and projective reductions and to study the projective plane and line). In these two papers - papers that have received, in our opinion, too little attention - the authors described a Hamiltonian structure obtained by reduction from a Poisson structure on $N$ copies of $G$ (or $G^{(N)})$ to a quotient $\N^{(N)}/B^{(N)}$ where $\N\subset G$ and the subgroup $B$ are chosen according to the finest gradation of the Lie algebra $\g$. (They work in the general semisimple case, while our paper works in the case $G = \PSL(n+1)$.) In this paper we show that by choosing a different gradation and 
enlarging both $\N$ and $B$ according to a certain parabolic choice, the quotient $\SL(n+1)^{(N)}/H^{(N)} \cong \N^{(N)}/B^{(N)}$ does not change and it represents the space of projective discrete invariants of $N$-gons in $\RP^n$. We then show that the Poisson bracket reduced in \cite{S2} can also be reduced with these new choices instead. This allows us to introduce several simplifications after which, surprisingly, one can see that the reduction coincides with that of the Sklyanin bracket associated to a very simple parabolic tensor (not an $R$-matrix). The new choice of quotient not only allow us to simplify the reduction, but it also allow us to connect it directly to invariant evolutions of projective $N$-gons.
 
In section 6 we prove that any evolution that is Hamiltonian with respect to the reduced Poisson bracket  has a geometric realization as invariant evolution of projective $N$-gons. We also provide a very simple relation between them, namely the gradient of the Hamiltonian needs to be equal to the coefficients of the discrete moving frame that defines the evolution of the polygons. 
This last section also describes completely integrable systems in the planar and higher dimensional cases. We first show that the integrable equation induced by
a certain invariant evolution of $N$-gons is the Boussinesq lattice related to the lattice $W_3$-algebra under a Miura transformation. This system appeared previously in { \cite{hiin97}}.
 In a final unexpected twist, we prove that the group right bracket reduces to a Poisson bracket in any dimension {\it even though it is also not Poisson prior to the reduction} (even if one chooses a valid $R$-matrix). We conjecture that both reductions form a Hamiltonian pencil for the integrable discretizations of the $W_n$-algebras and we prove this puzzling fact for the projective plane. To end the paper we define a generalization of our planar Boussinesq lattice to any dimension $n$ - a lattice discretization of the $W_n$-algebra - and we prove that it is completely integrable { by explicitly constructing its local master symmetry \cite{mr86c:58158}}. We also describe its rather simple projective realization.

From the work in \cite{S2} and \cite{S3} and the results in continuous cases there are clear indications that a similar study should be possible in a more general setting: that of parabolic manifolds and of geometric manifolds whose group of transformation is a a semidirect product of $\R^n$ with a semisimple group. These cases include most well-known flat geometries (Euclidean, conformal, Grassmannian, etc). It is remarkable to us that in all the examples we have studied discretizing the geometry produces integrable systems that are discretizations of known continuous integrable systems. This follows the main idea of modern discrete geometry to not discretize only the equation, but the entire geometric context. Given that many continuous integrable systems can be realized as geometric flows, this geometric process seems to give a path to finding integrable discretizations - not an trivial problem in itself - for flows that can be realized as geometric continuous flows. Work in these and other directions is 
currently underway. The authors would like to express their gratitude to Professor Semenov-Tian-Shansky for e-mail exchanges and assistance.

 \section{Background and definitions}
This section has two different, seemingly unrelated, parts. The first part deals with discrete group-based moving frames, the second with Poisson Lie-groups, classical $R$-matrices and the twisted quotient Poisson bracket.

\subsection{Discrete moving frames} 
 In this section we will describe basic definitions and facts needed along this paper on the subject of discrete group-based moving frames. They are taken from \cite{MMW} and occasionally slightly  modified to fit our needs. 

Let $M$ be a manifold and let $G\times M\to M$ be the action of a group $G$ on $M$. Although it is not needed, for simplicity we will assume from now on that $G$ is a subgroup of the general linear group. 
  \begin{definition}[Twisted $N$-gon]
{\it A twisted $N$-gon} in $M$ is a map $\phi:\Z\to M$ such that for some fixed $m\in G$ we have $ \phi(p+N) = m\cdot \phi(p)$ for all $p\in \Z$. (The notation $\cdot$ represents the action of $G$ on $M$.) The element $m\in G$ is called {\it the monodromy} of the gon. 
 \end{definition}
The main reason to work with twisted polygons is our desire to work with periodic invariants.
We will denote by $\P_N$ the space of twisted $N$-gons in $M$ and we will denote a twisted $N$-gon by its image $x = (x_s)$ where $x_s = \phi(s)$. If $G$ acts on $M$, it also has a natural induced action on $\P_N$ given by $g\cdot (x_s) = (g\cdot x_s)$.

 \begin{definition}[Discrete moving frame] Let $G^{(N)}$ denote the Cartesian product of $N$ copies of the group $G$. Allow $G$ to act on the left on $G^{(N)}$ using the diagonal action
\(
g\cdot (g_s) = (gg_s)
\)
(resp. right using the inverse diagonal action $g\cdot (g_s) = (g_sg^{-1})$.

 We say a map 
\[
\rho:\P_N \to G^{(N)}
\] 
is a left (resp. right)  {\it discrete moving frame} if $\rho$ is equivariant with respect to the action of $G$ on $\P_N$ and the left (resp. right inverse) { diagonal} action of $G$ on $G^{(N)}$. Since $\rho(x)\in G^{(N)}$, we will denote by $\rho_s$ its $s$th component; that is $\rho = (\rho_s)$, where $\rho_s(x) \in G$ for all $s$, $x = (x_s)$. Clearly, if $\rho = (\rho_s)$ is a left moving frame, then $\rho^{-1} = (\rho_s^{-1})$ is a right moving frame.
\end{definition}

Discrete moving frames are uniquely determined by a choice of {\it interlocking} transverse sections to the $G$-orbits on $M$. This was explained in \cite{MMW}. Indeed, if we choose sections $(T_s)$, $T_s \subset M$ transverse to the orbits of $G$, then a minimum number of conditions of the form
\begin{equation}\label{transverse}
g_p\cdot x_s \in T_s
\end{equation}
for some choices of $p$ and $s$ will determine $g_p$ completely in terms of $x$. These are called {\it normalization equations}. The authors of \cite{MMW} showed that if a choice of interlocking sections and normalization conditions determines $g = (g_s)$, then $\rho = g$ is a {\it right} discrete moving frame. One can choose normalization equations so that the associated moving frame is invariant under the {\it shift operator} $\T x_s = x_{s+1}$ (see \cite{MMW}).

 \begin{definition}[Discrete invariant] Let $F:\P_N \to \R$ be a function defined on $N$-gons. We say that $F$ is a scalar {\it discrete invariant} if
 \begin{equation}\label{invdef}
 F((g\cdot x_s)) = F((x_s))
 \end{equation}
 for any $g\in G$ and any $x = (x_s)\in \P_n$.
 \end{definition}
 We will naturally refer to vector discrete invariants when considering vectors whose components are discrete scalar invariants.

 \begin{definition}[Maurer-Cartan matrix] 
Let ${\rho}$ be a left (resp. right) discrete moving frame evaluated along a twisted $N$-gon. The element of the group
\[
K_s = \rho^{-1}_s\rho_{s+1} \hskip 2ex (\hbox{resp.}\hskip 1ex \rho_{s+1}\rho_s^{-1})
\]
is called the left (resp. right) $s$-Maurer-Cartan matrix for $\rho$. We will call the equation $\rho_{s+1} = \rho_s K_s$ the discrete $s$-Serret-Frenet equation. The element $K = (K_s) \in G^{(N)}$ is called the left (resp. right) {\it Maurer-Cartan matrix} for $\rho$.
\end{definition}
One can directly check that if $K = (K_s)$ is a left Maurer-Cartan matrix for the left frame $\rho$, then $(K_s^{-1})$ is a right one for the
right frame ${\rho^{-1}}=(\rho^{-1}_s)$, and vice versa. The entries of a Maurer-Cartan matrix are generators of {\it all discrete invariants}, as it was shown in \cite{MMW}. From now on we will assume, for simplicity, that $M = G/H$ is homogeneous and that $G$ acts on $M$ via left multiplication on representatives of the class. 

 Given a discrete {\it right} moving frame $\rho$, and assuming for simplicity that $\rho_s\cdot x_s = o$ for all $s$, one can describe the most general formula for an invariant evolution of polygons of the form
 \begin{equation}\label{invev0}
 (x_s)_t = F_s(x)
 \end{equation}
 in terms of the moving frame. This is reflected in the following theorem, which can be found in (\cite{MMW}). Denote by $\Phi_g:G/H\to G/H$ the map defined by the action of $g\in G$ on $G/H$, that is $\Phi_g(x) = g\cdot x$.
 
\begin{theorem}\label{evclass}
Any $G$-invariant evolution of the form (\ref{invev0}) can be written as
\begin{equation}\label{invev}
(x_s)_t = d\Phi_{\rho^{-1}_s}(o)(\v_s)
\end{equation}
where $o = [H]$, and $\v_s(x)\in T_{x_s} M$ is an invariant vector. 
\end{theorem}
If a family of polygons $x(t)$ is evolving according to (\ref{invev}), there is a simple process to describe the evolution induced on the Maurer-Cartan matrices, and hence on the invariants. It is described in the following theorem, which can be found at \cite{MMW} slightly modified.

 \begin{theorem}\label{structeq} Assume $x(t)$ is a flow of polygons solution of (\ref{invev}) and let $\s: G/H \to G$ be a section of $G/H$ such that $\s(o) = e \in G$. Let $\rho$ be a  right moving frame with $\rho_s\cdot x_s = o$ and assume that $\rho_s =\rho_s^H \s(x_s)^{-1}$, for some $\rho_s^H \in H$. Then
\begin{equation}\label{invarev}
(K_s)_t = N_{s+1}K_s -K_s N_s
\end{equation}
where $K_s$ is the right Maurer-Cartan matrix and $N_s =(\rho_s)_t\rho_s^{-1} \in \g$. Furthermore, assume $\g = \m\oplus \h$, where $\g$ is the algebra of $G$, $\h$ is the algebra of $H$ and $\m$ is a linear complement (which can be identified with the tangent to the image of the section $\s$). Then, if $N_s = N_s^\h + N_s^\m$ splits accordingly,
\begin{equation}\label{Ncond}
N_s^\m = - d\s(o) \v_s.
\end{equation}
\end{theorem}
As we will see in our next section, in the projective case equation (\ref{invarev}) and condition (\ref{Ncond}) completely determine the element $N = (N_s)$ with no need to know the moving frame explicitly.

\subsection{Poisson-Lie groups and the twisted quotient structure}
In this section we will assume that $G$ is semisimple and that there exists a nondegenerate invariant inner product $\langle, \rangle$ in $\g$ that allows us to identify $\g$ and $\g^\ast$. In the case of $G = \SL(n+1)$ the inner product is given by the trace of the product of matrices (with a factor of $1/2$ in the diagonal entries) so that  $E_{i,j}^\ast = E_{j,i}$, if $i\ne j$ and $(E_{i,i}-E_{n+1, n+1})^\ast = E_{i,i}-E_{n+1, n+1}$. The matrix $E_{i,j}$ has zeroes everywhere except of the $(i,j)$ entry, where it has a $1$. The following definitions and descriptions are taken from \cite{S2} and \cite{GS}.

\begin{definition}[Poisson-Lie group] A Poisson-Lie group is a Lie group equipped with a Poisson bracket such that the multiplication map $G\times G\to G$ is a Poisson map, where we consider the manifold $G\times G$ with the product Poisson bracket.
\end{definition}

\begin{definition}[Lie bialgebra] Let $\g$ be a Lie algebra such that $\g^\ast$ also has a Lie algebra structure given by a bracket $[,]_\ast$. Let $\delta: \g \to \Lambda^2\g$ be the dual map to the dual Lie bracket, that is
\[
\langle \delta(v), (\xi \wedge\eta) \rangle = \langle [\xi, \eta]_\ast, v\rangle
\]
for all $\xi, \eta \in \g^\ast$, $v\in \g$. Assume that $\delta$ is a one-cocycle, that is 
\[
\delta([v,w]) = [v\otimes{\bf 1} + {\bf 1}\otimes v, \delta(w)] - [w\otimes{\bf 1} + {\bf 1}\otimes w, \delta(v)]
\]
for all $v,w\in \g$. Then $(\g, \g^\ast)$ is called a {\it Lie bialgebra}.
\end{definition}
 If $G$ is a Lie-Poisson group, the linearization of the Poisson bracket at the identity defines a Lie bracket in $\g^\ast$ and $\delta$ is called the {\it cobracket}. The inverse result (any Lie bialgebra corresponds to a Lie-Poisson group) is also true for connected and simply connected Lie groups, a result due to Drinfel'd (\cite{D}).
 
The following definition will be used to prove our reduction theorem.
 
 \begin{definition}[Admissible subgroup] \label{admissible} Let $M$ be a Poisson manifold, $G$ a Poisson-Lie group and $G\times M\to M$ a Poisson action. A subgroup $H \subset G$ is called {\it admissible} if the space $C^\infty(M)^H$ of $H$-invariant functions on $M$ is a Poisson subalgebra of $C^\infty(M)$.
 \end{definition}
 
 The following proposition describes admissible subgroups.
 
  \begin{proposition}\label{admissibleprop} (\cite{S1}) Let $(\g, \g^\ast)$ be the tangent Lie bialgebra of a Poisson Lie group $G$. A Lie subgroup $H\subset G$ with Lie algebra $\h \subset \g$ is admissible if $\h^0\subset \g^\ast$ is a Lie subalgebra, where $\h^0$ is the annihilator of $\h$.
 \end{proposition}
  We will now describe Poisson brackets that will be central to our study.
 \begin{definition}[Factorizable Lie bialgebras and $R$-matrices] A Lie bialgebra $(\g, \g^\ast)$ is called {\it factorizable} if the following two conditions hold:
 \begin{enumerate}
 \item $\g$ is equipped with an invariant bilinear form $\langle, \rangle$ so that $\g^\ast$ can be identified with $g$ via $\xi\in \g^\ast \to v_\xi\in g$ with $\xi(\cdot) = \langle v_\xi, \cdot\rangle$;
 \item the Lie bracket on $\g^\ast\cong\g$ is given by
 \begin{equation}\label{dualbracket}
 [\xi, \eta]_\ast = \frac12 \left([R(\xi), \eta] + [\xi, R(\eta)]\right),
 \end{equation}
 where $R\in \mathrm{End}(\g)$ is a skew-symmetric operator satisfying the {\it modified classical Yang-Baxter equation}
 \[
 [R(\xi), R(\eta)] = R\left([R(\xi), \eta]+[\xi, R(\eta)]\right)-[\xi, \eta].
 \]
 \end{enumerate}
 $R$ is called  a {\it classical $R$-matrix}.  Let $r$ be the $2$-tensor image of $R$ under the identification $\g\otimes\g \cong \g\otimes\g^\ast \cong \mathrm{End}(\g)$. The tensor $r$ is often referred to as the $R$-matrix also. 
\end{definition}
Following \cite{S2}, we will consider factorizable Lie algebras satisfying the following two extra conditions:
\begin{enumerate}
\item there exists a linear map $r_+:\g^\ast \to \g$ such that both $r_+$ and $r_- = -r_+^\ast$ are Lie algebra homeomorphisms;
\item \label{dos} the endomorphism $t = r_+-r_-$ is $\g$-equivariant and {\it induces a linear isomorphism $\g^\ast \to \g$}.
\end{enumerate} 
Assume that $\g$ has a gradation of the form $\g = \g_+ \oplus \g_0 \oplus \g_-$, where $\g_+$ and $\g_-$ are dual of each other and where $\g_0$ is commutative (for example, when $\g_0$ is the Cartan subalgebra). Then it is well-known, and a simple straightforward calculation, that the map $R: \g \to \g$
\begin{equation}\label{R}
R(\xi_++\xi_0+\xi_-) = \xi_+ - \xi_-
\end{equation}
defines a classical $R$-matrix. Still, it fails to satisfy (\ref{dos}). In the projective case this can be remedied by adding a carefully chosen $\g_0$ perturbation. This $\g_0$ perturbation is crucial in the reduction in \cite{S2} and \cite{S3}. On the other hand, we will see that when reduced to the same quotient with a different representation, any such perturbation  vanishes and hence it plays no role in our paper. 

Given a Poisson Lie group and its associated Lie bialgebra, we can define a similar structure on $G^{(N)}$ (this was explained in \cite{S1}). Indeed, we equipped $\g^{(N)} = \displaystyle \bigoplus_N \g$ with a nondegenerate inner product given by 
\[
\langle X, Y\rangle = \sum_{k=1}^N \langle X_k, Y_k\rangle
\]
and we extend $R \in \mathrm{End}(\g)$ to $R\in \mathrm{End}(\g^{(N)})$ using $R((X_s)) = (R(X_s))$.  
Then $G^{(N)}$ is a Poisson Lie-group (with the product Poisson structure) and $(\g^{(N)}, \g^{(N)}_R)$ is its tangent Lie bialgebra, where $\g_R$ denotes $\g$ with Lie bracket (\ref{dualbracket}). 

The Poisson structure in $G^{(N)}$ is called the {\it Sklyanin bracket}, but that is not the bracket we are interested in to start with, although we will come back to it. Indeed, given a factorizable Lie bialgebra, the author of \cite{S2} and \cite{S1} defined what is called a {\it twisted Poisson structure on $G^{(N)}$}. In what follows we will give the definition of this structure and we refer the reader to \cite{S1} for explanations on how to obtain it, and to \cite{S2} (Theorem 1) for the explicit formula. (In an unexpected twist, we will show that the reduction of the twisted quotient Poisson structure with a certain choice of $R$-matrix and that of the Sklyanin bracket with a different - simpler - tensor coincide, but this does not seem clear a-priori.)

\begin{definition}[Left and right gradients] \label{gradients}
Let $F: G^{(N)} \to \R$ be a differentiable function. We define the {\it left gradient} of $F$ at $L = (L_s)\in G^{(N)}$ as the element of $\g^{(N)}$ denoted by $\nabla F(L) = (\nabla_sF(L))$ satisfying
\[
\frac d{d\epsilon}|_{\epsilon = 0} F((\exp(\epsilon \xi_s) L_s)) = \langle \nabla_s F(L), \xi_s\rangle
\]
 for all $s$ and any $\xi \in \g^{(N)}$, and where, abusing notation, we are using $\langle, \rangle$ to denote both the inner product in $\g$ and $\g^{(N)}$.
 
 Analogously, we define the {\it right gradient} of $F$ at $L$ as the element of $\g^{(N)}$ denoted by $\nabla'F(L) = (\nabla_s'F(L))$ satisfying
 \[
\frac d{d\epsilon}|_{\epsilon = 0} F((L_s\exp(\epsilon \xi_s) )) = \langle \nabla_s' F(L), \xi_s\rangle
\]
for all $s$ and any $\xi \in \g^{(N)}$.
 \end{definition}
 
There is a clear relation between left and right gradients. Since 
$$F((L_s\exp(\epsilon \xi_s) )) = F((\exp(\epsilon L_s \xi_sL_s^{-1})L_s )),$$ one sees that $\langle \nabla_s' F(L), \xi_s\rangle = \langle\nabla_sF(L), L_s\xi_sL_s^{-1}\rangle$. From here we get 
 \begin{equation}\label{lrgrad}
 \nabla_s'F(L) = L_s^{-1} \nabla_s F(L) L_s
 \end{equation}
  
The twisted Poisson structure can be explicitly described as follows. Let $F, G:G^{(N)} \to \R$ be two Hamiltonian functions.  Let $\T$ be the shift operator $\T(X_s) = X_{s+1}$. We define the $\T$-twisted Poisson bracket as
\begin{equation}\label{twisted}
\begin{array}{c}\{F, G\}(L) = \sum_{s=1}^Nr(\nabla_s F \wedge \nabla_s G) + \sum_{s=1}^Nr(\nabla_s' F \wedge \nabla_s' G) \\ \\- \sum_{s=1}^N (\T\otimes \mathrm{id})(r)(\nabla_s' F \otimes \nabla_s G) + \sum_{s=1}^N(\T\otimes\mathrm{id})(r)(\nabla_s' G \otimes \nabla_s F). \end{array}
\end{equation}
The authors of {\cite{S2, S1}} proved that, not only is this a Poisson bracket, but  the gauge action of $G^{(N)}$ in itself, that is,  the action $G^{(N)}\times G^{(N)} \to G^{(N)}$
\[
(L_s) \to (g_{s+1} L_s g_s^{-1}),
\]
is a Poisson map and the gauge orbits are Poisson submanifolds. In section \ref{sec5} we will prove that this bracket can be reduced to the space of projective invariants of twisted $N$-gons to produce a Hamiltonian structure which is naturally linked to projectively invariant evolutions of twisted gons. The reduction coincides with that of the Sklyanin bracket
\begin{equation}\label{Sklyanin}
\{F, G\}_S(L) = \sum_{s=1}^N\hat r(\nabla_s F \wedge \nabla_s G) - \sum_{s=1}^N\hat r(\nabla_s' F \wedge \nabla_s' G)
\end{equation}
for a different, simpler, choice of tensor $\hat r$.
\section{Projective moving frames and their Maurer-Cartan matrices}
We now turn our attention to the study of projective discrete moving frames. The readers familiar with the definition of Wilczynski's projective curve invariants (\cite{W}) will recognize the process below as its discrete analogue.

First of all we will describe the gradation of $\sl(n+1)$ that defines $\RP^n$ as parabolic manifold. We can write $\sl(n+1) = \g_1\oplus \g_0\oplus\g_{-1}$, where $\xi_i\in \g_i$, $i=1,0,-1$ are of the form
\begin{equation}\label{gradation}
\xi_1 = \begin{pmatrix} 0&\dots&0&0\\ \vdots&\vdots&\vdots&\vdots\\ 0&\dots&0&0\\ \ast&\dots &\ast&0\end{pmatrix}, \xi_0 = \begin{pmatrix} \ast&\dots&\ast&0\\ \vdots&\vdots&\vdots&\vdots\\ \ast&\dots&\ast&0\\ 0&\dots&0&\ast\end{pmatrix}, \xi_{-1} = \begin{pmatrix} 0&\dots&0&\ast\\ \vdots&\vdots&\vdots&\vdots\\ 0&\dots&0&\ast\\ 0&\dots&0&0\end{pmatrix}
\end{equation}

In what follows we will consider $\RP^n$ as the homogeneous space $\PSL(n+1)/H$, where $\PSL(n+1) = \SL(n+1)/\pm I$ is the projective linear simple group and $H$ is the subgroup corresponding to the Lie subalgebra $\g_1\oplus \g_0$. We will also consider the section of the quotient $\s:\RP^n \to \PSL(n+1)$ given by
\begin{equation}\label{section}
\s (x) = \begin{pmatrix} I_n&x\\ 0&1\end{pmatrix}
\end{equation}
where $x$ are the homogeneous affine coordinates in $\RP^n$ associated to the lift $x \in \RP^n \to (x, 1)\in\R^{n+1}$, and where $I_n$ is the $n\times n$ identity matrix. 
 
Assume $x=(x_s)$ is a twisted  $N$-gon in $\RP^n$ and consider the lift to $\R^{n+1}$ given by 
\begin{equation}\label{V}
V_s = t_s \begin{pmatrix} x_s \\ 1\end{pmatrix}.
\end{equation}
We will say a polygon is {\it nondegenerate} if $V_s, \dots, V_{s+n}$ are independent for all $s$. The following result is known but we could not find a simple reference so we are including a short proof. 

\begin{lemma}\label{invertible} Consider the $N\times N$ matrix
\[
A = \begin{pmatrix} 1 &  \dots& 1 & 0 &0&\dots & 0\\ 0&1&\dots & 1& 0&\dots&0\\ 0&0&1&\dots&1&\dots&0\\ \vdots&\ddots&\ddots&\ddots&\ddots&\ddots&\vdots\\\vdots&\ddots&\ddots&\ddots&\ddots&\ddots&\vdots\\ 1&\dots &0&\dots &0&1&1\\1&\dots &1&0&\dots &0&1\end{pmatrix}
\]
where we have a total of $k$ consecutive $1$'s in each row. Then A is invertible if and only if $N$ and $k$ are coprimes.
\end{lemma} 
\begin{proof} Let $E$ be the matrix
\[
E = \begin{pmatrix} 0&1&0&\dots&0\\ 0&0&1&\dots &0\\ \vdots&\ddots&\ddots&\ddots&\vdots\\ 0&0&\dots&0&1\\ 1&0&\dots&0&0\end{pmatrix}.
\]
Once can easily see that
\begin{equation}\label{expansion}
A = I + E + E^2 + \dots  E^{k-1}.
\end{equation}
Furthermore, the characteristic polynomial for $E$ is $\lambda^N -1$ so that its eigenvalues are $\omega^p$, $p = 0, \dots, N-1$, where $\omega = e^{\frac{2\pi i}N}$ is an $N$th root of unity. Thus, both E and A can be diagonalized, showing that the eigenvalues of A are given by
\[
\alpha_p = 1 + \omega^p+\omega^{2p} \dots+ \omega^{p(k-1)}
\]
$p=0,1,\dots,N-1$. Assume that $\alpha_p = 0$ for some $0<p<N$. Then $\omega^p$ is also solution of the equation $\lambda^k-1 = 0$ since $\alpha_p(\omega^p-1) = \omega^{pk} - 1 = 0$. Therefore, $\frac pN = \frac qk$ for some integer $q=0,1,\dots,k-1$. This occurs if, and only if $N$ and $k$ are not coprime.
\end{proof}
 The following proposition was proved in \cite{OST} for the case $n = 2$. 
\begin{proposition} {\it If $N$ and $n+1$ are coprime}, then $t_s$ can be found so that
\begin{equation}\label{normcond}
\det (V_{s+n}, \dots, V_{s+1}, V_s) = 1
\end{equation}
for any $s$. 
\end{proposition}
\begin{proof} Substituting (\ref{V}) in (\ref{normcond}) results in a system of equations of the form 
\[
t_{s+n} t_{s+n-1} \dots t_s = f_s
\]
for $s = 1, \dots N$, where 
\[
f_s^{-1} = \det(x_{s+n}-x_{s+n-1}, x_{s+n-1}-x_{s+n-2},\dots, x_{s+1}-x_s).
\]
We need to solve for $t_i$ using this system of equations. Applying logarithms changes the system into
\[
T_{s+n}+ T_{s+n-1} + \dots T_s = F_s
\]
where $T_i = \log t_i$ and $F_i = \log f_i$. Therefore, to prove the proposition we need to prove that the matrix of coefficients of this linear system is invertible. The coefficient matrix is the one displayed in lemma \ref{invertible} for $k = n+1$, ending the proof.
\end{proof}
Once we have defined the appropriate lifts, the definition of the left projective moving frame is rather simple.  Indeed, the map $(x_s) \to (\tilde\rho_s^{-1}) $, where
\begin{equation}\label{projmf}
\tilde\rho^{-1}_s = (V_{s+n}, \dots, V_{s+1}, V_s) \in \SL(n+1)
\end{equation}
defines a {\it left} discrete moving frame for $x=(x_s)$ since it is equivariant (we are using the inverse notation since we will reserve $\rho$ for right moving frames). Furthermore, it is straightforward to see that $\tilde\rho^{-1}_s \cdot o = x_s$ since the lift of $o\in G/H$ is the vector $e_{n+1} \in \R^{n+1}$ and $\tilde\rho^{-1}_s e_{n+1} = V_s$ whose projectivization is $x_s$. 

Once we have the left moving frame, the definition of the {\it left} Maurer-Cartan matrix is straightforward. Indeed, since $\{V_{s+n}, \dots, , V_{s+1}, V_s\}$ generates $\R^{n+1}$, one can always write a relation of the form
\[
V_{s+n+1} = \hat k_s^n V_{s+n} + \dots + \hat k_s^1 V_{s+1} + (-1)^n V_s
\]
where the coefficient of $V_s$ is determined by the condition $\det(V_{s+n}, \dots, V_s) = 1$ for all $s$ (and in particular for $s+1$). Thus, the {\it left} Maurer-Cartan equation is given by 
\begin{equation}
\tilde\rho^{-1}_{s+1} = \tilde\rho^{-1}_s \begin{pmatrix} \hat k_s^n & 1 & 0&\dots &0\\ \hat k_s^{n-1} & 0 & 1& \dots &0\\ \vdots& \vdots&\ddots&\ddots&\vdots\\ \hat k_s^1 & 0&\dots&0&1\\ (-1)^n & 0&\dots&0&0\end{pmatrix}.
\end{equation}
The corresponding {\it right} Maurer-Cartan matrix (generated by $\tilde\rho_s$) is given by the inverse of the matrix above, namely by
\[
\tilde K_s = \begin{pmatrix} 0&0&\dots&0&(-1)^n\\ 1&0&\dots &0&(-1)^{n-1}\hat k_s^n \\ 0&1&\dots&0&(-1)^{n-1} \hat k_s^{n-1}\\ \vdots&\ddots&\ddots&\ddots&\vdots\\ 0&\dots &0&1& (-1)^{n-1} \hat k_s^1\end{pmatrix} = \begin{pmatrix} 0&0&\dots&0&(-1)^n\\ 1&0&\dots &0&\tilde k_s^n \\ 0&1&\dots&0&\tilde k_s^{n-1}\\ \vdots&\ddots&\ddots&\ddots&\vdots\\ 0&\dots &0&1& \tilde k_s^1\end{pmatrix}.
\]
For reasons that will become clear later on, this is not yet our choice of invariants, moving frame and Maurer-Cartan matrix. The final choice is described in our next proposition.

\begin{proposition} There exists an invariant map $h: \RP(n+1)^N \to H^{(N)}$, (i.e., $h(x) = (h_s(x))\in H^{(N)}$ such that $h_s(g\cdot x) = h_s(x)$ for all $g\in \SL(n+1)$) gauging $\tilde K_s$  to
\begin{equation}\label{Ks}
K_s = \begin{pmatrix} 0& 0&\dots & 0&(-1)^n\\ 1&0&\dots &0&0\\ \vdots&\ddots&\ddots&\ddots&\vdots\\ 0&\dots &1&0&0\\k_s^2 & \dots &k_s^n&1& k_s^1\end{pmatrix}.
\end{equation}
That is, such that $K_s = h_{s+1} \tilde K_s h_s^{-1}$.
\end{proposition}  

Notice that this proposition implies the existence of a right discrete moving frame $\rho$ such that $\rho_s\cdot x_s = o$ ($h\in H^{(N)}$, which is the isotropy subgroup of $o$), with $K = (K_s)$ as its Maurer Cartan  matrix. Indeed, if $(\tilde\rho_s)$ is the right discrete moving frame generating $(\tilde K_s)$, then $(\rho_s) = (h_s \tilde \rho_s)$. 
\begin{proof}  If we choose
\[
h_s = \begin{pmatrix} 1 & 0 & 0&\dots &0\\ (-1)^{n-1}\tilde k_{s-1}^n & 1 & 0&\dots & 0\\ \vdots & \ddots & \ddots &\vdots\\ (-1)^{n-1} \tilde k_{s-1}^2 & 0 &\dots &1&0\\ 0&0&\dots &0&1\end{pmatrix}
\]
then 
\[
h_{s+1} \tilde K_s h_s^{-1}= \begin{pmatrix} 0&0&\dots &0&(-1)^n\\ 1&0&\dots &0&0\\ (-1)^n\tilde k_{s-1}^n&1&0&\dots&0\\ \vdots&\ddots&\ddots&\ddots&\vdots\\ (-1)^n \tilde k_{s-1}^2&0&\dots&1& \tilde k_s^1\end{pmatrix}
\]
Furthermore, if a matrix is of the form
\[
M_s = \begin{pmatrix} 0&0&0&0&\dots&0&0&0&(-1)^n\\ 1&0&0&0&\dots&0&0&0&0\\0&1&0&0&\dots&0&0&0&0\\ \vdots&\ddots&\ddots&\ddots&\dots&\dots&\vdots\\ 0&\dots&0&1&0&\ddots&0&0&0\\  0&\dots&0&a_s^n&1&\ddots&0&0&0\\ \vdots&\dots&\dots&\vdots&\ddots&\ddots&\ddots&\vdots\\ 0&\dots&0&a_s^{k+1}&0&\ddots&1&0&0\\a_s^2&\dots&a_s^{k-1}&a_s^k&0&\dots&0&1&a_s^1\end{pmatrix},
\]
then, after gauging it by the matrix 
\[
h_s = \begin{pmatrix} 1 & 0&0&\dots&0&0&0&0\\ 0&1&0&\dots&0&0&0&0\\ \vdots &\ddots&\ddots&\ddots&\ddots&\ddots&\vdots&\vdots\\ 
0&\dots&0&1&0&\dots&0&0\\ 0&\dots&0&-a_{s-1}^n&1&\dots&0&0\\ \vdots&\ddots&\ddots&\vdots&\ddots&\ddots&\vdots&\vdots\\ 0&\dots&0&-a_{s-1}^{k+2}&0&\dots&1&0\\ 0&\dots&0&0&0&\dots&0&1\end{pmatrix}
\]
it transforms into 
\[
h_{s+1} M_s h_s^{-1} =  \begin{pmatrix} 0&0&0&0&\dots&0&0&0&0&(-1)^n\\ 1&0&0&0&\dots&0&0&0&0&0\\ \vdots&\ddots&\ddots&\ddots&\dots&\dots&\ddots&\ddots&\vdots&\vdots\\ 0&\dots&0&1&0&0&\ddots&0&0&0\\  0&\dots&0&0&1&0&\ddots&0&0&0\\  0&\dots&0&0&a_{s-1}^n&1&0&\ddots&0&0\\ \vdots&\dots&\vdots&\vdots&\ddots&\ddots&\ddots&\ddots&\vdots&\vdots\\ 0&\dots&0&0&a_{s-1}^{k+2}&\ddots&0&1&0&0\\a_s^2&\dots&a^{k-1}_s&a_s^k&a_{s-1}^{k+1}&\dots&0&0&1&a_s^1\end{pmatrix}.
\]
We can reiterate this process $n-2$ times to finish the proof of the theorem.
\end{proof}
If we again describe $\sl(n+1) = \g_1\oplus \g_0\oplus \g_{-1}$ as in (\ref{gradation}), where the Lie subalgebra of $H$ is $\h = \g_1\oplus \g_0$ and where $\g_{-1}$ represents the tangent to the section $\s(\RP^n)$ at the identity, then, under this gauge, 
\(
(K_s)_t K_s^{-1} \in \g_1\) for all $s$. This is indeed the main advantage of adopting this choice of Maurer-Cartan matrix, as it will readily allows us to relate the invariant vector $\v_s$ in (\ref{invev}) to the variational derivative of a Hamiltonian function. (This will be explained in section \ref{sec6}.) To finish with this section, we will describe the formula (\ref{invev}) in our new frame $\rho$ and will initiate our running example. Since $\rho^{-1}_s = \tilde\rho_s^{-1} h_s^{-1}$ where $h_s$ is of the form
\[
h_s = \begin{pmatrix} 1& 0 &0&\dots&0\\ h_{21}&1&0&\dots&0\\ \vdots&\ddots&\ddots&\ddots&\vdots\\ h_{n 1}&\dots&h_{n n-1}&1&0\\ 0&\dots&0&0&1\end{pmatrix}
\]
the new left moving frame will be given by
\begin{equation}\label{finalframe}
\rho_s^{-1} = (V_{s+n}, \dots, V_{s+1}, V_s) h_s^{-1} = (W_{s+n},\dots, W_{s+1}, W_s)
\end{equation}
where $W_{s+k} = V_{s+k} + \sum_{i=1}^{k-1} h^{n-i+1, n-k+1} V_{s+i}$, $k\ge 1$ and $W_s = V_s$. Therefore, using (\ref{invev}), a general invariant evolution of polygons can be described in this new frame as the projectivization of the lifted evolution
\begin{equation}\label{lift}
(V_s)_t = \begin{pmatrix}W_{s+n} & \dots &W_{s+1}&W_s\end{pmatrix}\begin{pmatrix} \v_s\\ v^0_s\end{pmatrix}
\end{equation}
for some invariant vector $\v_s$ and with $v^0_s$ chosen uniquely in terms of $\v$ and the invariants so that the normalization condition (\ref{normcond}) is preserved.
\begin{example}\end{example}
 Our running example will be the projective plane. The case $\RP^1$ case was studied in (\cite{MMW}). In the planar case our first left moving frame is given by
 \(
 \tilde\rho_s^{-1} = (V_{s+2}, V_{s+1}, V_s) \) with $V_{s+3} = \hat a_s V_{s+2}+\hat b_sV_{s+1}+V_s$. The left Maurer-Cartan matrix is
 \[
 \hat K_s = \begin{pmatrix} \hat a_s & 1 & 0\\ \hat b_s & 0 & 1\\ 1&0&0\end{pmatrix}
 \]
 and the right one its inverse
 \[
 \tilde K_s =  \begin{pmatrix} 0&0&1\\ 1&0&-\hat a_s\\ 0&1&-\hat b_s\end{pmatrix}.
 \]
 (The reader should not be confused with the notation in \cite{OST}, where $\hat a$ and $\hat b$ are represented by $a$ and $b$.) We can gauge this matrix using 
 \[
 h_s = \begin{pmatrix} 1&0&0\\ \alpha_s & 1 & 0\\ 0&0&1\end{pmatrix} 
 \]
 as
 \[
 \begin{pmatrix} 1&0&0\\ \alpha_{s+1}&1&0\\ 0&0&1\end{pmatrix} \begin{pmatrix} 0&0&1\\ 1&0&-\hat a_s\\ 0&1&-\hat b_s\end{pmatrix} \begin{pmatrix} 1&0&0\\ -\alpha_s&1&0\\ 0&0&1\end{pmatrix}
 = \begin{pmatrix} 0&0&1\\ 1&0&-\hat a_s+\alpha_{s+1}\\ -\alpha_s&1&-\hat b_s\end{pmatrix}
 \]
 which, after choosing $\alpha_s = \hat a_{s-1}$ becomes
 \[
 K_s = \begin{pmatrix} 0&0&1\\ 1&0&0\\ a_s & 1 & b_s\end{pmatrix}
 \]
 where $a_s = -\hat a_{s-1}$ and $b_s = -\hat b_s$.
The new left moving frame is given by
\[
\rho_s^{-1} = \tilde\rho_s^{-1} h_s^{-1} = (V_{s+1}+a_s V_{s+1}, V_{s+1}, V_s)
\]
and an invariant evolution of polygons in this frame will be given by the projectivization of
\begin{equation}\label{planarev}
 (V_s)_t = v^1_s (V_{s+2}+a_s V_{s+1}) + v^2_s V_{s+1} + v^0_s V_s
 \end{equation}
 where $v^1$ and $v^2$ are arbitrary invariants and $v^0$ is uniquely determined by the normalization of the lift.

\section{An explicit formula for the evolution of the invariants}
We now turn to the following question: if we have an invariant evolution of the form (\ref{invev}), where $\rho_s$ is the (right) moving frame associated to $K_s$, how can we effectively obtain the explicit evolution of the invariants $k_s$? A partial answer is given by Theorem \ref{structeq}. If we choose the section $\s$ as in (\ref{section}),  we can write $N(x)=(N_s)\in \sl(n+1)^{(N)}$, with $N_s = (\rho_{s})_t\rho_s^{-1}$, as 
\begin{equation}\label{Ns}
N_s = \begin{pmatrix} A_s & \v_s\\ \a^T_s & -\mathrm{tr}(A_s) \end{pmatrix}
\end{equation}
where $\v_s$ is given by (\ref{invev}). The matrix $N_s$ will satisfy the structure equations
\begin{equation}\label{streq}
(K_s)_t = N_{s+1} K_s - K_s N_s.
\end{equation} 
Luckily, these equations also allow us to solve for $A_s$ and $\a_s$ for all $s$ so that $N_s$ is completely determined by $\v$, $\kb$ and (\ref{streq}).  

\begin{theorem}\label{Ndet} Assume the Maurer-Cartan element $(K_s)$ is defined by (\ref{Ks}), and assume $N_s$ is defined as in (\ref{Ns}). Assume further than the operator 
\begin{equation}\label{hypo}
\T+1+\T^{-1} + \dots + \T^{-(n-1)}
\end{equation}
is invertible. Then, the structure equations (\ref{streq}) determine uniquely $N_s$ and the evolution in time of the invariants $k_s^i$, $s = 1, \dots, N$, $i = 1,\dots n$ as functions of $\v$ and $\kb$.
\end{theorem}

\begin{proof}
Let us write $K_s$ and $N_s$ as
\[
\begin{pmatrix}\Lambda & (-1)^n e_1 \\ \overline{\kb_s}^T+e_n^T & k_s^1 \end{pmatrix}, \hskip 3ex N_s =\begin{pmatrix} A_s & \v_s\\ \a_s^T & -\tr(A_s)\end{pmatrix},
\]
where $\overline{\kb_s}^T = ( k_s^2, k_s^3, \dots, k_s^n, 0)$, and where $A_s$ and $\a_s$ are still to be found.  
With this notation the structure equations can be written as
\[
\begin{pmatrix} 0&0\\ (\overline{\kb_s}^T)_t & (k_s^1)_t\end{pmatrix} = 
\]
\[
\begin{pmatrix}\begin{array}{c} A_{s+1}\Lambda-\Lambda A_s +(-1)^{n-1}e_1 \a_s^T\\+ \v_{s+1}(\overline{\kb_s}^T+e_n^T)\end{array}&\begin{array}{c} (-1)^n (A_{s+1} e_1 +\tr(A_s) e_1) \\+ k_s^1 \v_{s+1}-\Lambda\v_s\end{array}\\\\ \begin{array}{c}\a_{s+1}^T\Lambda - k_s^1 \a_s^T \\-(\overline{\kb_s}^T+e_n^T)(A_s+\tr(A_{s+1}) I)\end{array}&\ast\end{pmatrix}.
\]
This equality implies conditions
\begin{eqnarray}
\label{eq1}A_{s+1}\Lambda - \Lambda A_s + \v_{s+1}(\overline{\kb_s}^T+e_n^T) + (-1)^{n+1}e_1 \a_s^T&=& 0\\\label{eq2} (-1)^n A_{s+1} e_1 + k_s^1\v_{s+1}-\Lambda \v_s + (-1)^n\tr(A_s)e_1 &=& 0.
\end{eqnarray}
and it describes $(k_s^i)_t$ in terms of $A_s$ and $\a_s$, for all $s$ and $i$. The first row of equation (\ref{eq1}) gives us
\[
e_1^TA_{s+1}\Lambda + \v_{s+1}^T e_1(\overline{\kb_s}^T+e_n^T) + (-1)^{n+1} \a_s^T = 0
\]
which will solve for $\a_s$ once $A_s$ has been found. The $p$ column of (\ref{eq1}), $p\ne n$, is given by
\begin{equation}\label{eq3}
A_{s+1} e_{p+1} - \Lambda A_s e_p + \v_{s+1} k_s^{p+1} +(-1)^{n+1} e_1 \a_s^T e_p = 0.
\end{equation}
This will solve for $A_s e_p$, except for the last entry, in terms of $A_{s+1} e_{p+1}$, for all $p\ne n$. Next we notice that the last column of (\ref{eq1}) is given by
\begin{equation}\label{eq4}
-\Lambda A_s e_n + \v_{s+1} + (-1)^{n+1} \a_s^T e_n = 0,
\end{equation}
which solves for all entries of $A_se_n$ with the exception of its last entry. Denote $A_s = (a_{i,j}^s)$. Using (\ref{eq3}) and (\ref{eq4}) we have solved for $a_{i,n}^s$, $i=1,\dots, n-1$, and using (\ref{eq3}) we have solved for $a_{i,j}$ with $i<j$, and we have found the recursion
\begin{equation}\label{rec}
a_{i, j}^s = a_{i+1,j+1}^{s+1} + \v_{s+1}^T e_{i+1}k_s^i
\end{equation}
whenever $i\ge j$.

Finally, we use (\ref{eq2}), whose entries other than the first one solve for $a_{i, 1}$, $i = 2,3,\dots n$. These entries and the recursion above determines $a_{i, j}^s$ for all $i>j$. As a last step, the first entry of (\ref{eq2}) gives us
\begin{equation}\label{eq5}
(-1)^n a_{1,1}^{s+1} + k_s^1\v_{s+1}^T e_1 + (-1)^n\tr(A_s) = 0.
\end{equation}
But according to the recursion (\ref{rec}), $a_{i,i}^s = a_{1,1}^{s-i+1} + F_i^s$, where $F_i^s$ is an expression depending on $\v$ and ${\bf k}$. Therefore
\[
\tr(A_s) = \sum_{i=1}^n a_{i,i}^s = \sum_{i=1}^n a_{1,1}^{s-i+1} + F_s .
\]
 Substituting this relation in equation (\ref{eq5}) we finally have
\[
(\T+1+\T^{-1}+\dots +\T^{-(n-1)}) a_{1,1}^s = G_s
\]
where, again, $G_s$ is a function depending on $\v$ and ${\bf k}$. Using our last hypothesis,  $\T+1+\T^{-1}+\dots+\T^{-(n-1)}$ is invertible so that we can finally solve for $a_{1,1}^s$, and with it all other entries of $N_s$.
\end{proof}

Hypothesis (\ref{hypo}) is not very restrictive. In fact, the operator $\T+1+\T^{-1}+\dots+\T^{-(n-1)}$ is invertible whenever $N$ and $n+1$ are coprime since, as a linear map, it is represented by the matrix $A$ in lemma (\ref{invertible}) - possibly up to some row exchanges.

Although in the general case the explicit expression of the evolution is too involved to be displayed here, working out a particular example can be done algebraically in an algorithmic fashion.
\begin{example} \end{example} 
 When $n=2$ the Maurer Cartan matrix and the matrix $N_s$ are given by
\begin{equation}\label{Ks2}
K_s = \begin{pmatrix} 0&0&1\\ 1&0&0\\ a_s & 1 & b_s\end{pmatrix}, \hskip 2ex N_s = \begin{pmatrix} A_s & B_s & \alpha_s \\ C_s&D_s& \beta_s\\ E_s&F_s&-A_s-D_s\end{pmatrix},
\end{equation}
where $\v_s = \begin{pmatrix} \alpha_s \\ \beta_s\end{pmatrix}$. From now on, and for simplicity's sake, we will denote $a_s$ merely by $a$ and $a_{s+p} = \T^p a$, for any $p$. Likewise with other functions. The structure equations, written as $(K_s)_tK_s^{-1} = N_{s+1}- K_s N_s K_s^{-1}$, become
\[
\begin{pmatrix} 0&0&0\\ 0&0&0\\ (b)_t & (a)_t & 0\end{pmatrix} 
\]
\[
=
\begin{pmatrix} 
(\T+1) A  + D + bF&\T B - E +aF & \T \alpha - F\\ \\
\T C - \alpha + bB &\T D - A + a B&\T \beta - B\\\\ \begin{array}{c}\T E - \beta - a\alpha \\+ b(A + 2D+aB+bF)\end{array} & \begin{array}{c}\T F - C - bE \\+ a(D-A+aB+bF)\end{array} & \ast
\end{pmatrix}.
\]
These equations completely determine the entries of $N_s$ to be given by
\[
\begin{array}{lll} F = \T \alpha & C = \T^{-1} \alpha - \T^{-1} b \T \beta & A = (\T + 1 +\T^{-1})^{-1}(\T^{-1}a\T \beta - b\T \alpha)\\ B = \T \beta & E= \T^2 \beta+a\T \alpha & D = -(\T+1+\T^{-1})^{-1}((\T^{-1}+1)a\T \beta +\T^{-1} b\T \alpha).
\end{array}
\]
They also determine the evolution of $b_s$ and $a_s$. This evolution can be written as
\begin{equation}
\begin{pmatrix} a \\ b\end{pmatrix}_t = \P \begin{pmatrix} \T \beta\\ \T \alpha\end{pmatrix}
\end{equation}
where
\begin{equation}\label{P}
\P = \begin{pmatrix}\begin{array}{c}\T^{-1} b - b\T\\+a(\T-\T^{-1})(\T+1+\T^{-1})^{-1}a \end{array}& \begin{array}{c}\T-\T^{-2}\\ +a(1-\T^{-1})(\T+1+\T^{-1})^{-1} b\end{array}\\\\ \begin{array}{c}\T^2-\T^{-1}\\-b(1-\T)(\T+1+\T^{-1})^{-1} a\end{array} &\begin{array}{c}\T a-a\T^{-1} \\+ b(\T-\T^{-1})(\T+1+\T^{-1})^{-1}b\end{array}\end{pmatrix}.
\end{equation}
Later, in section \ref{sec5}, we will show that $\P$ is a Poisson tensor for any dimension $n$. The tensor (\ref{P}) appeared in \cite{Ian}.
\vskip 2ex

\section{The Projective Hamiltonian structure on the space of invariants}\label{sec5}
In this section we aim to describe a naturally defined Poisson structure on the space of Maurer-Cartan matrices, and to give a precise account of how to obtain such a structure explicitly. The structure will be obtained via a reduction process from the twisted quotient structure on the Poisson-Lie group $\SL(n+1)^{(N)}$. The first step in our reduction process is to describe the space of Maurer-Cartan matrices as a quotient of the Poisson-Lie group $\SL(n+1)^{(N)}$ and to find explicitly the gradients of  functional extensions that are constant of the leaves of the quotient.
\subsection{The space of invariants as a quotient space}\label{quotient}

Assume we have a nondegenerate twisted polygon $x=(x_s)$ in a manifold $M = G/H$ with associated right moving frame $\rho$ such that $\rho_s\cdot x_s = o$ for all $s$. 

The subgroup $H^{(N)}$ acts naturally on $G^{(N)}$ via the gauge transformation 
\[
(g_s) \to (h_{s+1}g_sh_s^{-1} )
\]
(assuming $h_{s+N} = h_s$) and it is natural to ask what the quotient $G^{(N)}/H^{(N)}$ represents. The following result is valid not only for $\RP^n$, but for any homogeneous space.

\begin{theorem}
Locally around a nondegenerate polygon, a section of the quotient $G^{(N)}/H^{(N)}$ is given by the right Maurer-Cartan matrix $K$ associated to the right moving frame $\rho$. That is, let $x\in G^{(N)}/H^{(N)}$ be a nondegenerate twisted polygon, $\U$ an open set of $G^{(N)}/H^{(N)}$ containing nearby nondegenerate twisted polygons to $x$, and let $\K$ be the set of all the Maurer-Cartan matrices in $G^{(N)}$ associated to right moving frames along elements in $\U$ and determined by a fixed transverse section as in (\ref{transverse}). Then the map
\begin{equation}\label{map}
\K \to G^{(N)}/H^{(N)}, \hskip 3ex (K_s) \to [(K_s)]
\end{equation}
is a section of the quotient, a local isomorphism.
\end{theorem}
\begin{proof} Let $\U$ be a neighborhood of a nondegenerate polygon $x$ in $G^{(N)}/H^{(N)}$, small enough to preserve the non-degeneracy, and let's fix sections as in (\ref{transverse}) uniquely determining right moving frames for $x$ so that, $\rho_s\cdot x_s = o$.  Clearly, the map (\ref{map}) is well defined and 1-to-1, so we simply need to show that it is continuous and its image is an open set. 

Assume $M  \in G^{(N)}$ is nearby $K$ for some Maurer-Cartan matrix $K\in G^{(N)}$. We will show that $M$ can be gauged to a Maurer-Cartan element $\hat K$ corresponding to some polygon $\hat x$ nearby $x$. Define the recurrence relation
\[
\eta_{s+1} = M_s \eta_s
\]
for some $\eta_0$ fixed, 
and let $\hat x_s$ be the polygon defined by the vertices $\eta_s^{-1} \cdot o = \hat x_s$. If $M$ is nearby $K$, then $\hat x$ will be nearby $x$. We can use the same transverse sections defining $\rho$ and $K$ to find the left Maurer-Cartan matrix $\hat K$ corresponding to $\hat x$. If $M_s$ is close enough to $K_s$ the equations can always be solved and we can find $\hat K_s$ and its moving frame $\hat \rho_s$ such that $\hat\rho_s\cdot  \hat x_s =o$. Finally denote by $h_s$ the element $h_s =\hat\rho_s \eta_s^{-1}$. Clearly $ h_s\cdot o = o$ and so $h_s \in H$ for all $s$. Also, $\hat\rho_s = h_s\eta_s$ and so
\[
K_s =  h_{s+1} M_sh_s^{-1}
\]
which implies $[(K_s)] = [(M_s)]$. This concludes the proof.
\end{proof}

\subsection{Extensions constant on the $H^{(N)}$-gauge leaves} Let  $f:\K\to \R$ be a differentiable function on $\K\subset G^{(N)}$, viewed as a section of the quotient $G^{(N)}/H^{(N)}$. Assume $\F$ is an extension of $f$ to $G^{(N)}$ such that $\F$ is constant on the gauge leaves of $H^{(N)}$. That is, assume
\[
\F((h_{s+1}K_s h_s^{-1})) = f((K_s))
\]
for any $s$, any $h\in H^{(N)}$ and any $K\in \K$ as in (\ref{Ks}). In this section we aim to explicitly describe the left and right gradients of $\F$ evaluated along $\K$ in terms of the gradient of $f$ and the invariants $k_s^i$. 

\begin{proposition} 
Assume $f:\K\to \R$ is a function on $\K$, seen as a section of the quotient $\SL(n+1)^{(N)}/H^{(N)}$ given by (\ref{Ks}). Assume $\F$ is an extension of $f$ to $\SL(n+1)^{(N)}$. Then, the left gradient of $\F$ along $\K$ is given by
\[
\nabla_s \F(K) = \begin{pmatrix} Q_s &\frac{\partial f}{\partial \kb_s}\\ q_s^T & -\tr(Q_s)\end{pmatrix}
\]
where $\frac{\partial f}{\partial \kb_s} = \left((-1)^n \frac{\partial f}{\partial k^1_s}, \frac{\partial f}{\partial k^2_s},\dots, \frac{\partial f}{\partial k^n_s}\right)^T$.
\end{proposition}
\begin{proof}
Consider the element of $H$
\[
V_s = \begin{pmatrix}I_n & 0\\ \v_s^T &1\end{pmatrix} = \exp(\begin{pmatrix}0 & 0\\ \v_s^T &0\end{pmatrix})\in H
\]
where $\v_s = (v_s^i)$.  If $K_s$ is given as in (\ref{Ks}), then 
\[
V_s K_s = \begin{pmatrix} 0& 0&\dots & 0&(-1)^n\\ 1&0&\dots &0&0\\ \vdots&\ddots&\ddots&\ddots&\vdots\\ 0&\dots &1&0&0\\ k_s^2+v_s^2 & \dots &k_s^n+v_s^n&1& k_s^1+(-1)^n v_s^1\end{pmatrix}.
\]
Since $\F$ is an extension of $f$, it coincides with $f$ along $\K$ and so $\F((V_s K_s)) = f(k_s^1+(-1)^n v^1_s, k_s^2+v_s^2, \dots, k_s^n+v_s^n)$ . Differentiating we get
\[
\sum _{s=1}^N\langle \nabla_s\F, \begin{pmatrix}0 & 0\\ \v_s^T &0\end{pmatrix})\rangle = \sum_{s=1}^N\left( (-1)^n\frac{\partial f}{\partial k_s^1}v_s^1+\sum_{\ell=2}^n \frac{\partial f}{\partial k_s^\ell} v_s^\ell\right)
\]
and this is true for any values $v_s^i$. The proof of the proposition follows.
\end{proof}
The infinitesimal description of the fact that $\F$ is constant along the gauge leaves of $H$ is obtained by differentiating the relation $\F((h_{s+1}K_s h_s^{-1})) = \F((K_s))$ with $h = (h_s)=(\exp(t\xi_s))\in H^{(N)}$. This gives 
\[
\langle \nabla_s\F(K), \xi_{s+1}\rangle - \langle \nabla'_s\F(K), \xi_s\rangle = 0
\]
for all $\xi \in \h^{(N)}$. Which is the same as saying
\begin{equation}\label{streq2}
\T^{-1}\nabla_s\F - \nabla'_s\F \in \h^o = \g_1
\end{equation}
along $\K$. This property will determine the remaining entries of $\nabla_s \F(K)$.

\begin{theorem} Assume the Maurer-Cartan matrix $K$ is defined by (\ref{Ks}), and assume $\F$ is an extension of $f:\K\to \R$ to $\SL(n+1)^{(N)}$, constant on the gauge leaves of $H^{(N)}$. Assume further than the operator 
\[
\T+1+\T^{-1} + \dots + \T^{-(n-1)}
\]
is invertible. Then,  $\nabla \F = (\nabla_s\F)(K)$ is uniquely determined by (\ref{streq2}) as a function of the gradient of $f$ at $\kb = (\kb_s)$ and $\kb_s$.
\end{theorem}
\begin{proof} The proof of this theorem is almost identical to the proof of theorem \ref{Ndet}. From (\ref{streq2}), along $\K$ we have
\[
\T\nabla'_s\F -  \nabla_s\F  \in \h^o = \g_1.
\]
Notice also that $  \frac{\partial f}{\partial \kb_s}$  is in the $\g_{-1}$ position in $\nabla_s\F(K)$, and hence the $\g_{-1}$ entries in $\nabla'_s\F(K)$ are $\T^{-1}  \frac{\partial f}{\partial \kb_s}$. So is $\v_s$ for $N_s$. Furthermore, since
$(K_s)_t K_s^{-1} \in \g_1$, the structure equations (\ref{streq}) imply
\[
\T N_s  - K_s N_sK_s^{-1} \in \h^o=\g_1.
\]
 It suffices to choose $\T\nabla'_s\F$ in place of $N_s$ in the proof of theorem \ref{Ndet} to obtain the proof for our current theorem for $\T\nabla'_s\F$ and hence for $\nabla_s\F$.
\end{proof}
Here we see both the advantages of choosing a Maurer-Cartan matrix of the form (\ref{Ks}) and the anticipated relationship between the invariants coefficients $\v_s$ and the modified gradient of the Hamiltonian $f$, $\frac{\partial f}{\partial \kb_s}$.

To illustrate the process we work out our $\RP^2$ example.  

\begin{example}\end{example} In the planar case, recall that $K_s$ is given by (\ref{Ks2}). If we choose
\[
V_s = \begin{pmatrix} 1&0&0\\0&1&0\\\ w_s&v_s&1\end{pmatrix} = \exp\begin{pmatrix} 0&0&0\\0&0&0\\ w_s&v_s&0\end{pmatrix}  \in H
\]
we see  that
\[
V_sK_s = \begin{pmatrix} 0&0&1\\ 1&0&0\\ a_s+v_s&1&b_s+w_s\end{pmatrix}
\]
and, since $\F$ is an extension of $f$, they satisfy

\begin{equation}\label{relation1}
\F(V_sK_s) = f(a_s+v_s, b_s+w_s).
\end{equation}
 Let us write $f^s_a = \frac{\partial f}{\partial a_s}$ and $f_b^s =  \frac{\partial f}{\partial b_s}$. Relationship  (\ref{relation1}) implies 
\[
\langle\nabla_s\F(K), \begin{pmatrix} 0&0&0\\0&0&0\\ w_s&v_s&0\end{pmatrix}\rangle = \ha v_s + \hb w_s
\]
and so
\begin{equation}\label{grad2}
\nabla_s\F(K) = \begin{pmatrix} A_s& B_s & \hb\\ C_s&D_s& \ha\\ E_s&F_s&-(A_s+D_s)\end{pmatrix},
\end{equation}
where $A_s, B_s, C_s, D_s, E_s, F_s$ are unknown. As we saw in (\ref{lrgrad}), the right gradient can be obtained through the relation $\nabla'_s\F(K) = K_s^{-1}\nabla_s\F(K) K_s$. 

 If we substitute (\ref{grad2}) in (\ref{streq2}) we get the following expression. (Notice that, as before, we have dropped the subindices and denote $A_{s+p} = \T^{p} A$, etc.)
\begin{eqnarray*}
&&\T^{-1} \begin{pmatrix} A& B& \hbn\\ C&D& \han\\ E&F&-(A+D)\end{pmatrix}
\\
&-&\begin{pmatrix}
{D+a\han}&{\han}&{b\han+C}\\\\ \begin{array}{c}F-bB\\-a(A+2D+b\hbn+a\han)\end{array}&\begin{array}{c}-A-D\\-b\hbn-a\han\end{array}&\begin{array}{c}E-aC\\-b(2A+D+b\hbn+a\han)\end{array}\\\\ {B+a\hbn}&{\hbn}&{A+b\hbn}\end{pmatrix}\in \h^o=\g_1.
\end{eqnarray*}
Since $\g_1$ is given as in (\ref{gradation}), we obtain the following equations for the entries of $\nabla \F$
\[
\begin{array}{ccc} \T^{-1} B = \han& \T^{-1} A = D+a\han&\T^{-1}\han = E-aC-b(2A+D+b\hbn+a\han)\\ \T^{-1}\hbn = b\han+C& -\T^{-1}(A+D) = A+b\hbn& \T^{-1}C = F-bB-a(A+2D+b\hbn+a\han).
\end{array}
\]
Assuming that $\T+1+\T^{-1}$ is invertible (which we know is true as far as $N$ is not a multiple of 3), and denoting $(\T+1+\T^{-1})^{-1}$ by $\Ro$, these equations determine the entries of $\nabla F$ to be 
\begin{equation}\label{data}
\begin{array}{cc} B = \T\han& \begin{array}{c} E = \T^{-1}\han + a\T^{-1} \hbn \\+ b\Ro\left((1-\T)a\han + (\T^{-1}-\T)b\hbn\right)\end{array}\\ \\
C = \T^{-1}\hbn-b\han& \begin{array}{c} F = \T^{-2}\hbn+(b\T-\T^{-1}b)\han\\-a\Ro\left((\T-\T^{-1})a\han+(1-\T^{-1})b\hbn\right)\end{array}\\\\ A = \Ro(a\han-\T b\hbn)&D = -\Ro\left((1+\T)a\han+b\hbn\right)
\end{array}
\end{equation} 
\vskip 2ex

 We are now ready to move to the investigation of the Hamiltonian picture.
\subsection{Projective Hamiltonian structure on $\K$}

In this section we finally aim to prove that the twisted discrete Poisson bracket described in our initial section can be reduced to $\K \cong \SL(n+1)^{(N)}/H^{(N)}$, defining a natural Poisson bracket. We will work with the classical $R$-matrix given in \cite{S3}  associated instead to the finest gradation of $\sl(n+1)$; that is, associated to the splitting $\g_+\oplus\h_c\oplus\g_{-}$, where $\g_+$ are lower triangular matrices, $\g_-$ upper triangular ones, and $\h_c$ is the Cartan subalgebra.  We will not choose the particular $\h_c$ perturbation in \cite{S2}, but merely a more general one - any such choices will end up vanishing in our quotient. Thus, if $\xi = \xi_++\xi_c+\xi_{-}$ according to the gradation above,  consider the $r$-matrix defined as
\begin{equation}\label{rmatrix}
r(\xi, \eta) = \langle \xi_{-}, \eta_{+}\rangle + \frac12\sum_{p = 1}^N \phi_p \langle \xi_c, \T^p\eta_c\rangle
\end{equation}
where $\phi_p + \phi_{-p} = 2\delta_0^p$ can be any choice that will make $r$ an $R$-matrix.  

\begin{theorem} \label{5.5} The twisted Poisson structure (\ref{twisted}) defined on $\SL(n+1)^{(N)}$ with $r$ as in (\ref{rmatrix}) is locally reducible to the quotient $\SL(n+1)^{(N)}/H^{(N)}$. The $\g_0$-term in $r$   vanishes in the reduction and the reduced bracket coincides with the reduction of the Sklyanin bracket (\ref{Sklyanin}) with tensor
\[
\hat r(\xi, \eta) = \langle \xi_{-1}, \eta_1\rangle.
\]
\end{theorem}
Notice that $\hat r$ is not an $R$-matrix. 
\begin{proof} In order to prove this theorem we will apply the reduction theorem in \cite{MR}, reformulated according to our situation and notation. Their reduction theorem can be simplified as:
 
{\it Poisson reduction theorem}: Let $M$ be a Poisson manifold and assume that $E\subset TM$ is an integral and regular Hamiltonian subbundle (an integrable subbundle of the bundle defining the symplectic foliation). Assume $M/E$ is a manifold.  Then, if  the Poisson bracket preserves $E$ - that is, if the bracket of two functions constant on the leaves of $E$ is constant on the leaves of $E$ - the Poisson bracket can be reduced to the quotient $M/E$. The reduction is given by the formula
 \begin{equation}\label{relation}
 \{f, h\}_{M/E}([p]) = \{\F, \HH\}(p)
 \end{equation}
 where $\F, \HH: M\to\R$ are any extensions of $f, h:M/E\to \R$, constant on the leaves of $E$.
 \vskip 1ex
 In our case $M = G^{(N)}$ and the Poisson bracket is the twisted quotient bracket. The subbundle $E$ is the subbundle tangent to the $H$-gauge leaves, which is a Hamiltonian subbundle since the gauge leaves are Poisson submanifolds. This theorem is local and we know that $G^{(N)}/H^{(N)}$ is locally defined by $\K$, a manifold, and $E$ is regular when restricted to generic polygons. Therefore, to prove our theorem we simply need to check that $E$ preserves the twisted Poisson bracket. 
 
On the other hand, $E$ is defined by gauge orbits - a Poisson map for the twisted bracket according to \cite{S1} - and so it preserves the bracket whenever $H$ is admissible (see definition (\ref{admissible})). According to proposition \ref{admissibleprop}, this is true whenever $\h^0$ is a Lie subalgebra of $\g^\ast$, which is finally the only condition we need to check to prove the theorem.

Recall that the Lie bracket in $\g^\ast$ is defined by the linearization of the twisted Poisson bracket at the identity $e\in G$. That is
\[
[d_e\phi, d_e\varphi]_\ast = d_e\{\phi, \varphi\} \in \g^\ast.
\]
Since $\h^0 = \g_1$ we will look for functions $\phi^i_s$ such that $d_e\phi^i_s$ generate $\g_1$. Indeed, let $L\in G^{(N)}$ be close enough to $e\in G^{(N)}$ so that $L = (L_s)$ can be factored as
\[
L_s = \begin{pmatrix} I_n & \ell_s\\ {\bf 0}^T & 1\end{pmatrix} \begin{pmatrix} \Theta_s & {\bf 0} \\ {\bf 0}^T & \theta_s\end{pmatrix}\begin{pmatrix} I_n & {\bf 0}\\ q_s^T & 1\end{pmatrix}
\]
according to the gradation of the algebra. We define $\phi_s^i(L) = \ell^i_s$, where $i$ marks the $i$th entry. Instead of calculating $d_e\phi^i_s$ we will directly calculate the left gradient at $L$. Notice that if $g\in \SL(n+1)$,
\[
g \begin{pmatrix} I_n & \ell\\ {\bf 0}^T & 1\end{pmatrix} = \begin{pmatrix} I_n & g\cdot \ell\\ {\bf 0}^T & 1\end{pmatrix} g_H
\]
where $g\cdot \ell$ is the projective action of $\PSL(n+1)$ in $\R^n$ and $g_H\in H$. Therefore, 
\[
\phi^i_s (e^{\epsilon\xi} L) = (e^{\epsilon\xi}\cdot \ell_s)^i.
\]
We can now analyze each one of the factors. If $\xi \in \g_{-1}$ the action is simply a translation and so
\[
\nabla \phi_s^i(L) = \begin{pmatrix}\ast&\ast\\ e_i^T & \ast\end{pmatrix}.
\]
If $\xi\in \g_0$, the projective action of $e^{\epsilon\xi}$ is linear. Therefore 
\[
\nabla \phi_s^i(L) = \begin{pmatrix}A^i_s&\ast\\ e_i^T & -\ell_s^i\end{pmatrix}
\]
for $A_i^r = \begin{pmatrix} {\bf 0}&\dots&{\bf 0}&\ell_s&{\bf 0}&\dots&{\bf 0}\end{pmatrix}$, with the nonzero column located in the $i$th place. If $\xi\in \g_1$, then the infinitesimal projective action is quadratic and straightforward calculations show that
\[
\nabla \phi_s^i(L) = \begin{pmatrix}A^i_s&- \ell_s^i \ell_s\\ e_i^T & -\ell_s^i\end{pmatrix}.
\]
Whenever $L = e$, we have that $\ell_s = {\bf 0}$ and $d_e \phi_s^i = E_{n+1, i}$, generating $\g_1$. We know calculate $d_e\{\phi^i, \phi^j\}$ where $\{, \}$ is the twisted bracket (\ref{twisted}) with the $r$-matrix given by (\ref{rmatrix}). We want to show that $d_e\{\phi^i, \phi^j\} \in \h^0$ and so we need to show that $\frac d{d\epsilon}|_{\epsilon=0}\{\phi^i_s, \phi^j_s\}(e^{\epsilon\xi}) = 0$ whenever $\xi \in \h=\g_1\oplus\g_0$.

Notice that $(d_e\phi_s^i)_{-1} = 0$ and $(\nabla\phi_s^i(L))_{-1}$ is quadratic in $L$. Therefore we also have $\frac d{d\epsilon}|_{\epsilon=0}\left(\nabla\phi_s^i(e^{\epsilon\xi})\right)_{-1} = 0$.

Also, $\nabla' \phi_s^i(e^{\epsilon\xi}) = e^{-\epsilon\xi}\nabla \phi_s^i(e^{\epsilon\xi})e^{\epsilon\xi}$, and so
\[
\frac d{d\epsilon}|_{\epsilon=0} \nabla' \phi_s^i(e^{\epsilon\xi}) = [d_e\phi_s^i, \xi_s] + \frac d{d\epsilon}|_{\epsilon=0}\nabla\phi_s^i(e^{\epsilon\xi}).
\]
Since $d_e\phi_s^i \in \g_1$, whenever $\xi\in\h$ we have that $\frac d{d\epsilon}|_{\epsilon=0} \left(\nabla' \phi_s^i(e^{\epsilon\xi})\right)_{-1} = 0$. Furthermore, $\left(d_e\phi_s^i\right)_0 = 0$ also.
From here, and given that 
\[
\langle \nabla_+ \phi_s^i, \nabla_-\phi_s^i\rangle = \langle \nabla_1 \phi_s^i, \nabla_{-1}\phi_s^i\rangle+\langle \nabla_+^0 \phi_s^i, \nabla_-^0\phi_s^i\rangle,
\]
where $\nabla_+^0 \phi_s^i$ is the portion of $\nabla_+ \phi_s^i$ in $\g_0$, and similarly with the others, we get that $\frac d{d\epsilon}|_{\epsilon=0}\{\phi^i_s, \phi^j_s\}(e^{\epsilon\xi}) = 0$.

Next, we will show that the $\h_c$ portion of $r$ vanishes when reduced, and that the reduction of (\ref{twisted}) coincide with that of (\ref{Sklyanin}) associated to $\hat r$. From the definition, the terms in the bracket involving $\h_c$ terms of the gradients are given by a multiple of
\begin{eqnarray*}
&&\frac12\sum_{p=0}^{N-1} \phi_p(\langle (\nabla\F)_c, \T^p(\nabla\HH)_c\rangle - \langle (\nabla\HH)_c, \T^p(\nabla\F)_c\rangle \\&+& \langle (\nabla'\F)_c, \T^p(\nabla'\HH)_c\rangle - \langle (\nabla'\HH)_c, \T^p(\nabla'\F)_c\rangle)
\\&-&\sum_{p=0}^{N-1}\phi_p(\langle \T(\nabla'\F)_c, \T^p(\nabla\HH)_c\rangle - \langle \T(\nabla'\HH)_c, \T^p(\nabla\F)_c\rangle).
\end{eqnarray*}
Since extensions $\F$ and $\HH$ satisfy $\T^{-1}\nabla\F - \nabla'\F \in \g_1$, we have that $\T(\nabla'\F)_c = (\nabla \F)_c$ so that the above becomes
\begin{eqnarray*}
&&\frac12\sum_{p=0}^{N-1} \phi_p(-\langle (\nabla\F)_c, \T^p(\nabla\HH)_c\rangle + \langle (\nabla\HH)_c, \T^p(\nabla\F)_c\rangle \\&+& \langle (\nabla'\F)_c, \T^p(\nabla'\HH)_c\rangle - \langle (\nabla'\HH)_c, \T^p(\nabla'\F)_c\rangle)
\end{eqnarray*}
\begin{eqnarray*}
&=&\frac12\sum_{p=0}^{N-1} \phi_p(-\langle (\nabla\F)_c, \T^p(\nabla\HH)_c\rangle + \langle (\nabla\HH)_0, \T^p(\nabla\F)_c\rangle \\&+& \langle \T^{-1}(\nabla\F)_c, \T^{p-1}(\nabla\HH)_c\rangle - \langle \T^{-1}(\nabla\HH)_c, \T^{p-1}(\nabla\F)_c\rangle)=0.
\end{eqnarray*}
Finally, using the fact that $\T(\nabla'\F)_{-} = (\nabla\F)_{-}$ the reduced Poisson bracket can be expressed as
\begin{eqnarray*}
&&\{f, h\}({\bf k}) = \frac12 (\langle(\nabla \F)_{-}, (\nabla \HH)_+\rangle - \langle (\nabla \F)_+, (\nabla \HH)_{-}\rangle \\&+& \langle(\nabla' \F)_{-}, (\nabla' \HH)_+\rangle - \langle (\nabla' \F)_+, (\nabla' \HH)_{-}\rangle)
- \langle \tau(\nabla' \F)_{-}, (\nabla \HH)_+\rangle + \langle \tau(\nabla' \HH)_{-}, (\nabla \F)_+\rangle\\
 &=& \frac12\left(-\langle(\nabla \F)_{-}, (\nabla \HH)_+\rangle + \langle (\nabla \F)_+, (\nabla \HH)_{-}\rangle + \langle(\nabla' \F)_{-}, (\nabla' \HH)_+\rangle - \langle (\nabla' \F)_+, (\nabla' \HH)_{-}\rangle\right)
 \end{eqnarray*}
 \[
- \frac12\left(-\langle(\nabla \F)_{-}, (\nabla \HH)_+\rangle + \langle (\nabla \F)_+, (\nabla \HH)_{-}\rangle + \langle(\nabla' \F)_{-}, (\nabla' \HH)_+\rangle - \langle (\nabla' \F)_+, (\nabla' \HH)_{-}\rangle\right)\
 \]
\[
 = \frac12 \langle (\nabla \HH)_{-},  (\nabla \F)_+-\tau(\nabla'\F)_+ \rangle - \frac12 \langle (\nabla \F)_{-}, (\nabla \HH)_+ -\tau(\nabla'\HH)_+ \rangle.
\]
This is equal to
\[
 = \frac12 \langle (\nabla \HH)_{-1},  (\nabla \F)_{1}-\tau(\nabla'\F)_{1} \rangle - \frac12 \langle (\nabla \F)_{-1}, (\nabla \HH)_{1} -\tau(\nabla'\HH)_{1} \rangle.
\]
and from here we can go back to 
\[
- \frac12\left(-\langle(\nabla \F)_{-1}, (\nabla \HH)_1\rangle + \langle (\nabla \F)_1, (\nabla \HH)_{-1}\rangle + \langle(\nabla' \F)_{-1}, (\nabla' \HH)_1\rangle - \langle (\nabla' \F)_1, (\nabla' \HH)_{-1}\rangle\right),
 \]
which coincides with the reduction of (\ref{Sklyanin})associated to $\hat r$. Even more surprising, we will later show that the right bracket produces a Poisson bracket upon reduction, even though the original bracket is not Poisson!
 \end{proof}
Using (\ref{relation}) we can actually calculate explicitly the reduction of the twisted bracket to $\K$. We will illustrate it with our running $\RP^2$ example.
 \begin{example}\end{example} Assume $f:\K\to \R$ is a Hamiltonian function and let $\F:\SL(3)^{(N)}\to \R$ be an extension of $f$, constant on the gauge leaves of $H^{(N)}$. Recall that we have explicitly found the left gradient of such a extension along $\K$. It is given by (\ref{grad2}). We also know that the right and left gradients satisfy equation (\ref{streq2}), and so the $\g_1\oplus\g_0$ component of $\nabla'_s\F(K)$
 equals that of $\T^{-1}\nabla_s\F(K)$. We have
 \[
 \nabla \F(K) = \begin{pmatrix} A& B & \hbn\\ C&D& \han\\ E&F&-(A+D)\end{pmatrix},\hskip 2ex \nabla'\F(K) = \begin{pmatrix} \T^{-1}A& \T^{-1}B & \T^{-1}\hbn\\ \T^{-1}C&\T^{-1}D& \T^{-1}\han\\ B+a\hbn&\hbn&-\T^{-1}(A+D)\end{pmatrix},
 \]
 where we have again dropped the subindex to avoid over-cluttering, and where the values for $A, B, C, D, E, F$ were found in (\ref{data}). The reduced Poisson bracket can be obtained by simply substituting both gradients of the extensions in (\ref{twisted}). Consider  the now simplified tensor $\hat r$ 
 \[
\hat  r( P\otimes Q) =  \langle P_{-1}, Q_{1}\rangle.
 \]
Denote by $A^g$ the entry of $\nabla \G(K)$, likewise with other entries. As before, we will ignore subindices, and we will write
\[
\langle v, w \rangle = v^T w.
\]
Substituting in (\ref{twisted}) we obtain that the reduced bracket is given by
\begin{eqnarray*}
\{ f, g\}(\kb) &= &\frac12\langle\begin{pmatrix}\hbn\\ \han\end{pmatrix}, \begin{pmatrix} E^g \\ F^g\end{pmatrix}\rangle - \frac12 \langle\begin{pmatrix}g_b\\ g_a\end{pmatrix}, \begin{pmatrix} E^f \\ F^f\end{pmatrix}\rangle
\\&+&\frac12\langle\begin{pmatrix}\T^{-1}\hbn\\ \T^{-1}\han\end{pmatrix} , \begin{pmatrix} B^g + a g_b\\ g_b\end{pmatrix}\rangle - \frac12\langle\begin{pmatrix}\T^{-1}g_b\\ \T^{-1}g_a\end{pmatrix} , \begin{pmatrix} B^f + a \hbn\\ \hbn\end{pmatrix}\rangle\\ &-&\langle\begin{pmatrix}\hbn\\ \han\end{pmatrix}, \begin{pmatrix} E^g \\ F^g\end{pmatrix}\rangle +  \langle\begin{pmatrix}g_b\\ g_a\end{pmatrix}, \begin{pmatrix} E^f \\ F^f\end{pmatrix}\rangle
\end{eqnarray*}
\[
= \frac12\langle\begin{pmatrix}\hbn\\ \han\end{pmatrix}, \begin{pmatrix}\T B^g + \T ag_b- E^g \\ \T g_b-F^g\end{pmatrix}\rangle - \frac12 \langle\begin{pmatrix}g_b\\ g_a\end{pmatrix}, \begin{pmatrix} \T B^f+\T a\hbn -E^f \\ \T \hbn-F^f\end{pmatrix}\rangle
\]
\[
=\langle \begin{pmatrix} \han \\ \hbn \end{pmatrix},  \P \begin{pmatrix} g_a\\ g_b\end{pmatrix}\rangle
\]
where $\P$ is given as in (\ref{P}). 

\section{Hamiltonian evolutions of twisted polygons}\label{sec6}
By now it is clear that there is a very close relationship between the evolution induced on the invariants by invariant evolutions of polygons, and the Hamiltonian evolution associated to the reduced bracket obtained from the twisted bracket (\ref{twisted}). Indeed, we have seen in our example that they are equal under some identifications. In our final section we will prove that if we choose as invariant coefficients $\T \v_s = \frac{\partial f}{\partial \kb_s}$, then the evolution of the projective polygons defined by $\v_s$ induces a Hamiltonian evolution on $\kb_s$, with Hamiltonian function $f$. The result implies that {\it any $n$-dimensional reduced Hamiltonian evolution is induced on $\kb$ by some invariant evolution of projective polygons in $\RP^n$}.

\begin{theorem}\label{th61}
Assume an invariant evolution of twisted $N$-polygons in $\RP^n$ lifts to an evolution of the form (\ref{lift}). Furthermore, assume that
\begin{equation}\label{compa}
\T \v_s =  \frac{\partial f}{\partial \kb_s}
\end{equation}
for some function $f:\K\to \R$, where 
 \(
 \frac{\partial f}{\partial \kb_s} = ( (-1)^n\frac{\partial f}{\partial k^1_s},  \frac{\partial f}{\partial k^2_s},\dots,  \frac{\partial f}{\partial k^n_s})^T\). Then, the evolution induced on the invariants $k_s^i$ is the reduced Hamiltonian evolution associated to the Hamiltonian function $f$.
 \end{theorem}
 
We will refer to (\ref{compa}) as the {\it compatibility condition}. 

 \begin{proof} Assume $\kb$ evolves by an evolution that is Hamiltonian with respect to the reduced bracket on $\K$. Let's denote by $\xi_f$ the Hamiltonian vector field associated to a Hamiltonian function $f: \K \to \R$ such that $\kb_t = \xi_f(\kb)$. 
 
 Now, given  that $\kb_t$ appears in the $\g_1$ component of $K_t K^{-1}$ and $\frac{\partial f}{\partial \kb}$ is in the $\g_{-1}$ position of $\nabla \F(K)$ for any extension of $f$ constant on the leaves of $H$ as in (\ref{streq2}), we have that the reduced bracket of $f$ with any other function $g$ can be written as
 \[
\{f,g\}(\kb) =  \xi_f(\kb) (g) = \langle K_t K^{-1}, \nabla \G(K)\rangle
\]
where $\F$ and $\G$ are extensions of $f,g:\K\to \R$ as in (\ref{streq2}). Now, notice that if $\F$ and $\G$ satisfy (\ref{streq2}) we have
\begin{equation}\label{+-rel}
(\nabla \F)_{-1} = (\T\nabla'\F)_{-1}, \hskip 2ex\hbox{or}\hskip3ex (\T^{-1}\nabla \F)_{-1} = (\nabla' \F)_{-1}.
\end{equation}
We are assuming that all gradients are evaluated at $K$. Using this relation in the Sklyanin bracket  we obtain
\begin{eqnarray*}&&\{f, g\}(\kb) = \langle K_t K^{-1}, \nabla \G(K)\rangle\\&=& \frac12 \langle (\nabla \G)_{-1},  (\nabla \F)_1-\T(\nabla'\F)_1 \rangle - \frac12 \langle (\nabla \F)_{-1}, (\nabla \G)_1 -\T(\nabla'\G)_1 \rangle.
 \end{eqnarray*}
  But these expressions are skew-symmetric. Indeed, notice that  $(\nabla \G)_1-\T(\nabla'\G)_1 = \nabla \G -\T\nabla'\G$ since it belongs to $\g_1$ and, therefore,
  \[
   \langle (\nabla \F)_{-1}, (\nabla \G)_1 -\T(\nabla'\G)_1 \rangle =  \langle \nabla \F, \nabla \G -\T\nabla'\G \rangle = -\langle \nabla \F, \T\nabla'\G\rangle +\langle\T\nabla'\F, \T \nabla'\G\rangle
   \]
   where we have used that $\langle\nabla \F, \nabla \G\rangle = \langle\nabla' \F, \nabla' \G\rangle$ since $\langle, \rangle$ is invariant under the adjoint action. From here, the above equals
   \[
    \langle(\T\nabla'\G)_{-1}, (\T\nabla'\F)_1-(\nabla \F)_1\rangle =- \langle(\nabla \G)_{-1}, (\nabla \F)_1-(\T\nabla'\F)_1\rangle.
     \]
  Skew-symmetry tells us that
  \[
\{f, g\}(\kb) =  \langle K_t K^{-1}, \nabla \G(K)\rangle =      \langle (\nabla \F)_1-\T(\nabla'\F)_1 ,  (\nabla \G)_{-1}\rangle.
 \]
 We can now see the relation between $\T\v$ and $\frac{\partial f}{\partial \kb}$. Indeed, if $N$ is given as in (\ref{streq}), 
 \[
  \langle K_t K^{-1}, \nabla \G(K)\rangle =  \langle \T N - K N K^{-1}, \nabla \G(K)\rangle
  \]
 for any extension $\G$ as in (\ref{streq2}). Recall next that $T N - K N K^{-1}\in \g_1$, and both $\T N$ and $\nabla \F$ are determined by its $\g_{-1}$ component and this condition. Notice also that $\nabla \F = K \nabla' \F K^{-1}$. Thus, we can conclude that 
 \[
 N = -\nabla' \F.
 \]
 Since $\T N_{-1} =- (\T\nabla'\F)_{-1}  =- (\nabla \F)_{-1}$, and the $\g_{-1}$ component of $N$ is $-\v$, the theorem follows.
 \end{proof}
\subsection{Completely integrable evolutions of planar polygons}\label{sec61}
 In this section, we study in detail the integrable lattice that appears in the case of  planar polygons. Although this, and the next study, seem to separate themselves from invariant evolutions of polygons, we will return to connect them towards the end.

In Section \ref{sec5} we have shown that the operator $\P$ defined by (\ref{P}) is Hamiltonian, which naturally leads to Hamiltonian
evolutions for the invariants as stated in Theorem \ref{th61}. However, to obtain integrable systems we need biHamiltonian structures. 
To obtain one or more compatible structures, we shall introduce arbitrary constants in the operator $\P$ 
and study the conditions on the parameters to ensure that the operator is  still Hamiltonian. That is, we will analyze all possible Hamiltonian structures, compatible or not, within $\P$.
\begin{theorem}\label{th62} Consider the antisymmetric operator $\cH$ given by
\begin{eqnarray*}
&&\begin{pmatrix}\begin{array}{c}\la_1 (\T^{-1} b - b\T)\\+ \la_2 a(\T-\T^{-1})(\T+1+\T^{-1})^{-1}a \end{array}& 
\begin{array}{c}\la_4 \T-\la_3 \T^{-2}\\ +\la_5 a(1-\T^{-1})(\T+1+\T^{-1})^{-1} b\end{array}\\\\ 
\begin{array}{c}\la_3 \T^2-\la_4 \T^{-1}\\-\la_5 b(1-\T)(\T+1+\T^{-1})^{-1} a\end{array} 
&\begin{array}{c}\la_6 (\T a-a\T^{-1}) \\+ \la_7 b(\T-\T^{-1})(\T+1+\T^{-1})^{-1}b\end{array}\end{pmatrix},
\end{eqnarray*}
where $\la_i$, $i=1,\cdots, 7$ are constants. Then $\cH$ is Hamiltonian when one of the following three cases is satisfied
\begin{enumerate}
\item[(1).]  $ \la_1=\la_2=\la_5=\la_6=\la_7=0, \quad $ $\la_3$ and $\la_4$ are any constants;
\item[(2).] $ \la_1 \la_6=\la_2 \la_3,\ \la_3=\la_4,\ \la_2=\la_5=\la_7$ and at least one of $\la_1, \la_3, \la_6$ is nonzero;
\item[(3).] $\la_1=\la_3=\la_4=\la_6=0$, $\quad \la_2, \la_5$ and $ \la_7  $ are any constants.
\end{enumerate}
\end{theorem}
\begin{proof}
We investigate the conditions needed for $\cH$ to be a Hamiltonian operator using Proposition 7.7 in \cite{mr94g:58260}. 
Although this theorem is formulated for differential operators, it is also valid for difference operators (\cite{mwx2}).

For the operator $\cH$, we have
\begin{eqnarray*}
 \cH (\bt)=\cH \begin{pmatrix} \th\\ \et\end{pmatrix}=\begin{pmatrix}\la_1( b_{-1} \th_{-1} - b\th_1) 
+ \la _2 a( P_1-P_{-1})
+\la_4 \et_1-\la_3 \et_{-2} + \la_5 a(Q-Q_{-1})\\ 
{\la_3 \th_2-\la_4 \th_{-1}-\la_5 b(P-P_{1})+
 \la_6(a_1 \et_1-a \et_{-1}) +\la_7 b(Q_1-Q_{-1})}\end{pmatrix}
\end{eqnarray*}
where we use notations
\begin{eqnarray*}
(\T+1+\T^{-1})^{-1}a \th =P \quad \mbox{and} \quad (\T+1+\T^{-1})^{-1}b \et =Q,
\end{eqnarray*}
that is,
\begin{eqnarray*}
a \th =P_1+P+P_{-1} \quad \mbox{and} \quad b \et =Q_1 +Q+Q_{-1}.
\end{eqnarray*}
 (The reader should not confused the notation here with our previous one. By $P_1$ we mean $\T P$ and not its $\g_1$ component.)
We now define the following tri-vector
\begin{eqnarray*}
 \Psi=\frac{1}{2} \int \bt \wedge \mbox{Pr}_{\cH(\bt)} \cH \wedge \bt \ .
\end{eqnarray*}
We know that an anti-symmetric operator $\cH$ is Hamiltonian if and only if the tri-vector $\Theta$ vanishes (\cite{mr94g:58260}).
We carry out this computation for the operator $\cH$ defined in the statement.
\begin{eqnarray*}
 &&\Psi=\int\left(-\la_1 \th \wedge b' \wedge \th_{1} +\la_2 \th \wedge a'\wedge ( P_1-P_{-1}) +\la_5 \th \wedge a'\wedge (Q-Q_{-1})\right.\\
 &&\qquad \left.-\la_5\et \wdg b' \wdg (P-P_{1})-\la_6 \et \wdg a' \wdg \et_{-1} +\la_7 \et\wdg b' \wdg (Q_1-Q_{-1})
 \right).
\end{eqnarray*}
Here $a'$ takes the value of the first entry of $\cH (\bt)$ and $b'$ takes the value of its second entry. We substitute them into the above expression.
Instead of computing the whole expression of $\Theta$, we compute its independent terms, which should all vanish. First let us look at $3$-forms involving only $\th$.
These terms are
\begin{eqnarray*}
&&\!\!\!\int\!\!\!\left(-\la_1 \th \wedge (\la_3 \th_2-\la_4 \th_{-1}-\la_5 b(P-P_{1})) \wedge \th_{1} +\la_2 \th \wedge 
\la_1( b_{-1} \th_{-1} - b\th_1)\wedge ( P_1-P_{-1})\right)
 \\
&=&\int \left( \la_1 \la_4 \th \wedge  \th_{-1}\wedge \th_{1}- \la_1 \la_3 \th \wedge  \th_{2} \wedge \th_{1}
+\la_1 \la_5 b \th \wedge (P-P_{1}) \wedge \th_{1}\right.\\
&&\left. +\la_2 \la_1 b_{-1} \th \wedge \th_{-1} \wedge ( P_1-P_{-1}) -\la_2 \la_1 b\th \wedge  \th_1\wedge ( P_1-P_{-1})\right)\\
&=&\!\!\int\!\! \left( -\la_1 (\la_4-\la_3) \th \wedge  \th_1\wedge \th_{-1}
  +\la_1 b \th \wedge \th_{1} \wedge (\la_2( P-P_{1}+P_{-1}-P_{2})-\la_5 (P-P_{1}))\right)\\
&=&\la_1 \int \left(  (\la_4-\la_3) \th \wedge  \th_1\wedge \th_{-1}
+(\la_2-\la_5)  b \th \wedge \th_{1} \wedge (P-P_{1})\right).
\end{eqnarray*}
Here we used the identities $a \th =P_1+P+P_{-1}$ and $a_{1} \th_{1} =P+P_{1}+P_{2}$.
This part vanishes when 
\begin{eqnarray}\label{cond1}
 \la_1 (\la_4-\la_3)=0, \qquad \la_1 (\la_2-\la_5)=0.
\end{eqnarray}
Next let us look at $3$-forms involving only $\et$. These terms are
\begin{eqnarray*}
&&\int\left(-\la_6 \et \wdg (\la_4 \et_1-\la_3 \et_{-2} + \la_5 a(Q-Q_{-1})) \wdg \et_{-1} \right.\\
&&\left.+\la_7 \et\wdg (\la_6(a_1 \et_1-a \et_{-1}) ) \wdg (Q_1-Q_{-1})
 \right)\\
&=&\la_6 \int \left( (\la_3-\la_4) \et \wedge \et_{1}\wedge \et_{-1}-\la_5 a \et\wdg \et_{-1}  \wdg  (Q_{-1}-Q)\right.\\
&&\left. -\la_7 a \et\wdg \et_{-1}  \wdg (Q_1-Q_{-1}+Q-Q_{-2}) \right)\\
&=&\la_6 \int \left( (\la_3-\la_4) \et \wedge \et_{-1}\wedge \et_1+(\la_7-\la_5) a \et\wdg \et_{-1}  \wdg  (Q_{-1}-Q)  \right).
\end{eqnarray*}
Here we used the identities $b \et =Q_1+Q+Q_{-1}$ and $b_{-1} \et_{-1} =Q+Q_{-1}+Q_{-2}$.
These terms vanish when 
\begin{eqnarray}\label{cond2}
 \la_6 (\la_4-\la_3)=0, \qquad \la_6 (\la_7-\la_5)=0.
\end{eqnarray}
We then look at the terms involving a 2-form in $\th$ and a one-form in $\et$.
\begin{eqnarray*}
 &&\int\left(-\la_1 \th \wedge (\la_6(a_1 \et_1-a \et_{-1}) +\la_7 b(Q_1-Q_{-1})) \wedge \th_{1}\right.\\
 &&+\la_2 \th \wedge (\la_4 \et_1-\la_3 \et_{-2} + \la_5 a(Q-Q_{-1}))\wedge ( P_1-P_{-1}) \\
 &&+\la_5 \th \wedge (\la_1( b_{-1} \th_{-1} - b\th_1) + \la _2 a( P_1-P_{-1}))\wedge (Q-Q_{-1})\\
 && \left.-\la_5\et \wdg (\la_3 \th_2-\la_4 \th_{-1}-\la_5 b(P-P_{1})) \wdg (P-P_{1})
 \right)\\
&=&\!\!\int\!\!\left(a \la_1 \la_6  \th \wedge \et\wedge \th_{-1}-a\la_1 \la_6\th_1 \wedge\et_{-1} \wedge \th
-b (\la_7-\la_5) \la_1 \th \wedge(Q_1-Q_{-1}) \wedge \th_{1} \right.\\
&&+\la_2 \la_3 \th_{1} \wedge  \et_{-1} \wedge ( P-P_{2}) 
-\la_4\la_2 \th_{-1} \wedge  \et \wedge ( P_{-2}-P)\\
&& \left.+\la_5 \la_4\et \wdg \th_{-1} \wdg (P-P_{1})
-\la_3\la_5\et_{-1} \wdg  \th_{1} \wdg (P_{-1}-P)\right)\\
&=&\!\!\int\!\!\left(a \la_1 \la_6  \th \wedge \et\wedge \th_{-1}-a\la_1 \la_6\th_1 \wedge\et_{-1} \wedge \th
-b (\la_7-\la_5) \la_1 \th \wedge(Q_1-Q_{-1}) \wedge \th_{1} \right.\\
&&\!\!+\la_2 \la_3 \th_{1} \wedge  \et_{-1} \wedge ( 2P+P_1+P_{-1}\!-\!P_{-1}) 
\!-\!\la_4\la_2 \th_{-1} \wedge  \et \wedge (P_1\!-\!P_1\! -\!P_{-1}\!-\!2 P)\\
&& \left.+\la_5 \la_4\et \wdg \th_{-1} \wdg (P-P_{1})
-\la_3\la_5\et_{-1} \wdg  \th_{1} \wdg (P_{-1}-P)\right)\\
&=&\int\left(a (\la_1 \la_6-\la_4\la_2)  \th \wedge \et\wedge \th_{-1}-a(\la_1 \la_6-\la_2 \la_3)\th_1 \wedge\et_{-1} \wedge \th \right.\\
&&-b (\la_7-\la_5) \la_1 \th \wedge(Q_1-Q_{-1}) \wedge \th_{1} \\
&&\left.+(\la_2-\la_5) \la_3 \th_{1} \wedge  \et_{-1} \wedge ( P-P_{-1}) 
-\la_4 (\la_2-\la_5) \th_{-1} \wedge  \et \wedge (P_1- P)\right).
\end{eqnarray*}
These terms vanish if
\begin{eqnarray}\label{cond3}
\begin{array}{lll}\la_1 \la_6-\la_2 \la_3=0, & \la_1 \la_6-\la_4\la_2=0, & (\la_7-\la_5) \la_1=0,\\
(\la_2-\la_5) \la_3=0, & \la_4(\la_2-\la_5)=0 . &\end{array}
\end{eqnarray}
Finally, we look at the terms involving a 2-form in $\et$ and a one-form in $\th$.
\begin{eqnarray*}
 &&\int\left(\la_5 \th \wedge (\la_4 \et_1-\la_3 \et_{-2} + \la_5 a(Q-Q_{-1}))\wedge (Q-Q_{-1})\right.\\
 &&-\la_5\et \wdg (\la_6(a_1 \et_1-a \et_{-1}) +\la_7 b(Q_1-Q_{-1})) \wdg (P-P_{1})\\
 &&-\la_6 \et \wdg (\la_1( b_{-1} \th_{-1} - b\th_1) + \la _2 a( P_1-P_{-1})) \wdg \et_{-1}\\
 &&\left.+\la_7 \et\wdg (\la_3 \th_2-\la_4 \th_{-1}-\la_5 b(P-P_{1})) \wdg (Q_1-Q_{-1})
 \right)\\
&=&\!\!\int\!\!\left(\la_1 \la_6 b \et \wedge \th_{1}\wedge \et_{-1}\!-\!b\la_1 \la_6\et_1 \wedge\th \wedge \et
-a (\la_2-\la_5) \la_6 \et \wedge(P_1-P_{-1}) \wedge \et_{-1} \right.\\
&&+\la_5 \la_4 \th \wedge  \et_1 \wedge ( Q-Q_{-1}) 
-\la_5\la_3 \th_1 \wedge  \et_{-1} \wedge ( Q_1-Q)\\
&& \left.+\la_7 \la_3\et_{-1} \wdg \th_1 \wdg (Q-Q_{-2})
-\la_4\la_7\et_1 \wdg  \th \wdg (Q_2-Q)\right)\\
&=&\int\left( b( \la_1 \la_6 -\la_3\la_7) \et \wedge \th_{1}\wedge \et_{-1}-b (\la_1 \la_6-\la_4\la_7)\et_1 \wedge\th \wedge \et  \right.\\
&&-a (\la_2-\la_5) \la_6 \et \wedge(P_1-P_{-1}) \wedge \et_{-1} \\
&&\left.+(\la_5-\la_7) \la_4 \th \wedge  \et_1 \wedge ( Q-Q_{-1}) 
-(\la_5-\la_7)\la_3 \th_1 \wedge  \et_{-1} \wedge ( Q_1-Q)
\right).
\end{eqnarray*}
It vanishes when
\begin{eqnarray}\label{cond4}
\begin{array}{lll}\la_1 \la_6-\la_7 \la_3=0, & \la_1 \la_6-\la_4\la_7=0, &  (\la_2-\la_5) \la_6=0,\\
(\la_5-\la_7) \la_3=0, & (\la_7-\la_5) \la_4=0 . &\end{array}
\end{eqnarray}
From here we conclude that the operator $\cH$ is Hamiltonian if and only if the parameters satisfy all the conditions written out in
(\ref{cond1}), (\ref{cond2}), (\ref{cond3}) and (\ref{cond4}).
Using the Maple package Gr{\"o}bner to solve this algebraic system, we can sum up the solutions and obtain the three 
cases listed in the statement of the theorem.
\end{proof}
Clearly case (1) in Theorem \ref{th62} is not interesting in our search for integrable systems since the Hamiltonian pair is independent of dependent variables.
The same happens to the Hamiltonian operator in case (3) since we can rewrite it as 
\begin{eqnarray*}
 \begin{pmatrix} a&0\\0&b\end{pmatrix}\!\! \begin{pmatrix} \la_2 (\T-\T^{-1})(\T+1+\T^{-1})^{-1} \!\!&\! 
\la_5 (1-\T^{-1})(\T+1+\T^{-1})^{-1} \\
-\la_5 (1-\T)(\T+1+\T^{-1})^{-1}
\!\!&\! \la_7 (\T-\T^{-1})(\T+1+\T^{-1})^{-1}\end{pmatrix}\!\!
 \begin{pmatrix} a&0\\0&b\end{pmatrix}.
\end{eqnarray*}
We now look at case (2). Without losing generality, we can write down two Hamiltonian pairs as stated in the following theorem. 
For convenience, we write $a$ and $b$ as the dependent variables in one pair and $\tilde a$ and $\tilde b$ in another pair.
\begin{theorem}\label{th63}
Let 
$\la$ be an arbitrary constant. Then, the operators
\begin{eqnarray*}
&&\P_1=\begin{pmatrix}0 & \T- \T^{-2}\\  \T^2- \T^{-1}
&\la (\T a-a\T^{-1})\end{pmatrix}
\end{eqnarray*}
and
\begin{eqnarray*}
&&\P_2=\begin{pmatrix}\!\!\begin{array}{c}\T^{-1} b - b\T\\+ \la a(\T-\T^{-1})(\T+1+\T^{-1})^{-1}a \end{array}& 
\la a(1-\T^{-1})(\T+1+\T^{-1})^{-1} b\\\\ 
-\la b(1-\T)(\T+1+\T^{-1})^{-1} a 
& \la b(\T-\T^{-1})(\T+1+\T^{-1})^{-1}b\end{pmatrix}
\end{eqnarray*}
form a Hamiltonian pair, and the operators 
\begin{eqnarray*}
&&\Q_1(\tilde a, \tilde b)=\begin{pmatrix}\la(\T^{-1} \tilde b -\tilde b\T)& 
 \T- \T^{-2}\\ 
\T^2- \T^{-1}
&0\end{pmatrix}
\end{eqnarray*}
and
\begin{eqnarray*}
&&\!\!\Q_2(\tilde a, \tilde b)\!=\!\begin{pmatrix}\!\la \tilde a(\T\!-\!\T^{-1})(\T+1+\T^{-1}\!)^{-1} \tilde a\!\! & \!
\la \tilde a(1\!-\!\T^{-1})(\T+1+\T^{-1}\!)^{-1} \tilde b\\ \\
-\la \tilde b(1-\T)(\T+1+\T^{-1}\!)^{-1} \tilde a \!\!
&\!\begin{array}{c} \T \tilde a-\tilde a\T^{-1} \\+ \la \tilde b(\T\!-\!\T^{-1}\!)(\T+1+\T^{-1}\!)^{-1} \tilde b\end{array} 
\!\!\end{pmatrix}
\end{eqnarray*}
form another Hamiltonian pair. Moreover, the Hamiltonian pair $\P_1$ and $\P_2$ is related to the Hamiltonian pair $\Q_1$ and $\Q_2$ by the Miura transformation
$\tilde a= -b$ and $\tilde b=-a_1$, that is,
$$\Q_1(\tilde a, \tilde b)=D_{(\tilde a, \tilde b)} \P_1 D_{(\tilde a, \tilde b)}^\star 
\quad \mbox{and} \quad 
\Q_2(\tilde a, \tilde b)=D_{(\tilde a, \tilde b)} \P_2 D_{(\tilde a, \tilde b)}^\star .$$
This Miura transformation is induced by projective duality.
\end{theorem}
\begin{proof} The first part of the proof of this statement is straightforward. 
They are two Hamiltonian pairs as a direct result of case (2) in Theorem \ref{th62} and the relation between them can be
checked by a simple calculation. To see that the Miura transformation is induced by projective duality, we will use the map $\alpha$ defined in \cite{OST}.  (Recall that $(a,b)$ in \cite{OST} is in fact $(\hat a, \hat b)$ here,  a difference created by the gauge transformation we initially introduced and the choice of right Maurer-Cartan matrix rather than the left one \cite{OST} uses.) Define $\alpha:\RP^2\to (\RP^2)^\ast$ as
\[
\alpha(x_i)  = \overline{x_i x_{i+1}},
\]
where $\overline{x_i x_{i+1}}$ is the line joining $x_i$ and $x_{i+1}$. The authors of \cite{OST} showed that $\alpha^\ast(\hat a_i) = -\hat b_{i+1}$ and $\alpha^\ast(\hat b_i) = -\hat a_i$ where if $(\hat a_i, \hat b_i)$ are the invariants for $x_i$,
$(\alpha^\ast(\hat a_i), \alpha^\ast(\hat b_i))$ are the projective invariants associated to the polygon whose vertices are the points dual to $\alpha(x_i)$. Given that
\[ 
\hat a_i = -a_{i+1}, \hskip .5in \hat b_i = -b_i
\]
we have
\[
\alpha^\ast(a_i) = -b_i, \hskip 2ex \alpha^\ast(b_i) = -a_{i+1}.
\]
 \end{proof}

From here it follows that, to study integrable systems associated to biHamiltonian structures included in $\P$, we only need to derive biHamiltonian integrable systems for the Hamiltonian pair $\P_1$ and $\P_2$.
Based on a difference analogue to the Adler Residue Theorem (\cite{mr80i:58026,mwx2}), we conclude that the first Hamiltonian associated to this pair is 
$\ln b$. This leads to the following bi-Hamiltonian system
\begin{eqnarray}\label{eqc2}
&&\left\{\begin{array}{l} a_t=\frac{1}{b_1}-\frac{1}{b_{-2}} \\ b_t =\la \left(\frac{a_1}{b_1}-\frac{a}{b_{-1}}\right)
\end{array}\right.
\end{eqnarray}
with
\begin{eqnarray}\label{bihc2}
\begin{pmatrix} a \\ b\end{pmatrix}_t= \P_1\delta f =  \P_2 \delta g, \quad 
\mbox{where} \quad f=\ln b \quad \mbox{and} \quad g=-\frac{a_1}{b b_1}.
\end{eqnarray}
Notice that the constant $\la$ in (\ref{eqc2}) can be scaled away by a simple scaling transformation $a\mapsto \frac{a}{\la^{1/3}}$
and $b\mapsto \la^{1/3} b$ if $\la \neq 0$. We will keep it instead of taking $\la=1$.

Let us introduce the following transformation (\cite{hiin97})
\begin{eqnarray}\label{hiin}
u=\frac{1}{b b_1 b_2}, \qquad v=-\frac{a_1}{b b_1}.
\end{eqnarray}
Their Fr{\'e}chet derivatives with respect to $a$ and $b$ are given by
\begin{eqnarray*}
&&D_{(u,v)}=\begin{pmatrix}  0
& -u (1+\T+\T^2)\frac{1}{b}\\
v \T \frac{1}{a}& -v (1+\T) \frac{1}{b}\end{pmatrix}
=\begin{pmatrix}  0
& u (1+\T+\T^2)\\
v \T & v (1+\T) \end{pmatrix} 
\begin{pmatrix}  \frac{1}{a}
&0\\
0&  -\frac{1}{b}\end{pmatrix}.
\end{eqnarray*}
Therefore, when we change the variable the Hamiltonian operator $\P_1$ transforms into
\begin{eqnarray}
&&\tilde \P_1=D_{(u,v)} \P_1 D_{(u,v)}^\star\nonumber\\
&&=\begin{pmatrix}  0\!\!
& u (1+\T+\T^2)\\
v \T \!\!& v (1+\T) \end{pmatrix} 
\begin{pmatrix}  0 \!\!
& \frac{u_{-1}}{v_{-1}}\T\!-\!\T^{-2} \frac{u}{v_{1}}\\
\frac{u}{v_1}\T^2\!-\!\T^{-1} \frac{u_{-1}}{v_{-1}} \!\!&\la (\T^{-1} v\!-\!v\T) \end{pmatrix} \nonumber\\
&&\qquad \begin{pmatrix}  0 \!\!& \T^{-1} v\\ (1+\T^{-1}\!+\!\T^{-2}) u
\!\! & (1+\T^{-1}\!)v \end{pmatrix} 
\nonumber\\
&&=\begin{pmatrix}0
& u(1+\T+\T^2) \left( u \T -\T^{-2}u \right)\\
 \left(u \T^{2} -\T^{-1} u \right) (1+\T^{-1}+\T^{-2}) u & \begin{array}{c}
v(1+\T) (u \T-\T^{-2} u)\\+(u\T^2-\T^{-1}u)(1+\T^{-1})v\end{array}
\end{pmatrix}\nonumber\\
&&+\la\! \begin{pmatrix}\!  u (1+\T+\T^2)\!\!\!\!\!\!\!\!\!\!
& 0\\
0 \!\!\!\!\!\!\!\!\!\!& v (1+\T)\! \end{pmatrix} (\T^{-1}\! v\!-\!v \T) \begin{pmatrix}\!  1\!\!\!
& 1\\
1\!\!\!& 1 \!\end{pmatrix} 
\begin{pmatrix}\!  (1+\T^{-1}\!+\T^{-2}\!)u\!\!\!\!\!\!\!\!\!\!
& 0\\
0 \!\!\!\!\!\!\!\!\!\!&  (1+\T^{-1}\!)v\! \end{pmatrix}.
\nonumber
\end{eqnarray}
Similarly, the Hamiltonian operator $\P_2$ changes into
\begin{eqnarray*}
&&\tilde \P_2=\begin{pmatrix} \la u (\T-\T^{-1}) (\T+1+\T^{-1}) u & \la u (\T-1)(\T+1+\T^{-1}) v\\
\la v (1-\T^{-1}) (\T+1+\T^{-1}) u & \la v(\T-\T^{-1}) v +\T^{-1} u-u\T
\end{pmatrix}.\nonumber
\end{eqnarray*}
Notice that {\it under this transformation both operators are local}.  This Hamiltonian pair can be found in \cite{Belov}.
The biHamiltonian system (\ref{eqc2}) becomes
\begin{eqnarray}\label{equv}
\begin{pmatrix} u \\ v\end{pmatrix}_t=\begin{pmatrix} \la u (v_2-v_{-1}) \\ u_{-1}-u +\la v (v_1-v_{-1})\end{pmatrix}
=\tilde \P_1\delta f = \tilde \P_2 \delta g,
\end{eqnarray}
where $f=-\frac{1}{3} \ln u$ and $g=v.$ When $\la=1$ (as mentioned before, we can scale $\la$ to $1$ when $\la\neq 0$), 
this system has appeared in \cite{hiin97}, where the authors
studied the integrable systems related the lattice $W$-algebras. It is the Boussinesq lattice related to the
lattice $W_3$-algebra.

Notice that $\P_1\big{|}_{\la=1}+\P_2\big{|}_{\la=1}=\P$, where the operator $\P$ is defined by (\ref{P}). Notice also that $\P_2(\ln b) = 0$. Therefore, 
for the evolution induced on the invariants by invariant evolutions of planar polygons, we have the following 
result:
\begin{theorem} Let  $f(a,b) = \ln b$ and let $\pi: \R^3\to \RP^2$ be the projection associated to the lift $x\to \begin{pmatrix} x\\ 1\end{pmatrix}$. Then,
the evolution induced on the invariants $a$ and $b$ by the invariant evolution of planar polygons 
\[
(x_s)_t = \pi(\frac {1} {b_s} V_{s+2}+\frac {a_s}{b_s} V_{s+1}+ V_s)
\]
is the biHamiltonian equation
\begin{eqnarray}\label{evolab}
&&\begin{pmatrix} a \\ b\end{pmatrix}_t=
\begin{pmatrix} \frac{1}{b_1}-\frac{1}{b_{-2}} \\  \frac{a_1}{b_1}-\frac{a}{b_{-1}}\end{pmatrix}
= \P\delta \ln b =  \P_2\big{|}_{\la=1} \delta \left(-\frac{a_1}{b b_1}\right) .
\end{eqnarray}
Under the Miura transformation (\ref{hiin}), it is transformed into the Boussinesq lattice related to the
lattice $W_3$-algebra
\begin{eqnarray*}
\begin{pmatrix} u \\ v\end{pmatrix}_t=\begin{pmatrix}  u (v_2-v_{-1}) \\ u_{-1}-u +v (v_1-v_{-1})\end{pmatrix} .
\end{eqnarray*}
\end{theorem}
Based on this theorem, we can anticipate that integrable discretizations of $W_n$-algebras should be induced by invariant
evolutions of projective polygons in $\RP^n$. In our last section we produce the $\RP^n$ generalization and we prove it is completely integrable. But first we point at a further reduction  that produces yet another integrable system.

Indeed, the next symmetry flow of equation (\ref{evolab}) is
\begin{eqnarray*}
\begin{pmatrix} a\\b
\end{pmatrix}_{\tau}=\P_1\big{|}_{\la=1} \delta \left(-\frac{a_1}{b b_1}\right) =
\begin{pmatrix}
 \frac{a_2}{b_1^2 b_2}+\frac{a_1}{b b_1^2}-\frac{a_{-1}}{b_{-2}^2 b_{-1}}-\frac{a_{-2}}{b_{-2}^2 b_{-3}}\\
\frac{1}{b_{-1}b_{-2}}-\frac{1}{b_1 b_2}+\frac{a_1}{b_1^2} (\frac{a_2}{b^2}+\frac{a_1}{b})
-\frac{a}{b_{-1}^2} (\frac{a}{b}+\frac{a_{-1}}{b_{-2}})
\end{pmatrix}.
\end{eqnarray*}
Here we use time variable $\tau$ instead of $t$ to avoid the confusion.
It admits a reduction 
$$b_{\tau}=\frac{1}{b_{-1}b_{-2}}-\frac{1}{b_1 b_2},
$$
when $a=0$. Under the Miura transformation (\ref{hiin}), it becomes
$$u_{\tau}=u (u_1+u_2-u_{-1}-u_{-2}).$$
This is the Narita-Itoh-Bogoyavlensky lattice  \cite{Bog88}
$$u_{t}=u (\sum_{i=1}^p u_{i}-\sum_{i=1}^p u_{-i})
$$
with $p=2$. 
In general, the Narita-Itoh-Bogoyavlensky lattice 
is equivalent to the lattice $W_{p+1}$ algebra \cite{hiin97}. 
Recently, the biHamiltonian structures for the Narita-Itoh-Bogoyavlensky lattice
were constructed using the Lax representation \cite{w2012}. Notice that there is no reduction for the corresponding 
invariant evolution of planar polygons since
\begin{eqnarray*}
\P \delta \left(-\frac{a_1}{b b_1}\right) = (\P_1\big{|}_{\la=1}+\P_2\big{|}_{\la=1}) \delta \left(-\frac{a_1}{b b_1}\right) =\begin{pmatrix} a\\b
\end{pmatrix}_{\tau}+\begin{pmatrix} a\\b
\end{pmatrix}_{t}.
\end{eqnarray*}

Notice that the Poisson structure used in \cite{OST} is not compatible with the ones used here, a fact that was also pointed out in \cite{Ian}. Still, here we have a variety of brackets and one ideally would check the behavior of the pentagram map with respect to any of them.

\subsection{Hamiltonian pencils and completely integrable systems in $\RP^n$}

In our last section we wonder about the origins of the pencil in Theorem \ref{th63} and we try to generalize it, together with the planar integrable system. In particular, let us define the right bracket on $G^{(N)}$ as
\begin{equation}\label{rightbr}
\{\F, \G\}'(L) = \hat r\left(\nabla'\F(L)\wedge\nabla'\G(L)\right).
\end{equation}
Even if we were to use an $R$-matrix $r$ for its definition instead of $\hat r$, this (or the parallel left) bracket is in general not Poisson, and one can easily check  that, regardless of any $\h_c$ perturbation as in \cite{S2}, Jacobi's identity is not satisfied if $G = \SL(2)$, for example.  Still, we will next show that when it is reduced to $\K$, by evaluating it along the proper extensions, the resulting bracket is Poisson. In the particular case of $\RP^2$, we in fact obtain $\frac12 \P_1|_{\lambda = 1}$, which forms a Hamiltonian pencil with our original reduction. We work out this example first.
\begin{example} \end{example} Let $f, g:\K\to\R$ be two functions and let $\F, \G$ be extensions satisfying (\ref{streq2}). Then
\[
\{f,g\}_0(a,b) = \{\F, \G\}'(K) = \frac12\langle\begin{pmatrix} \T^{-1}f_b\\\T^{-1}f_a\end{pmatrix}, \begin{pmatrix} B_g+ag_b\\ g_b\end{pmatrix} \rangle - \frac12 \begin{pmatrix} \T^{-1}g_b\\\T^{-1}g_a\end{pmatrix}, \begin{pmatrix} B_f+af_b\\ f_b\end{pmatrix}
\]
\[
= \frac12\begin{pmatrix}f_a&f_b\end{pmatrix}\begin{pmatrix} 0&\T-\T^{-2}\\ \T^2-\T^{-1}& \T a-a\T^{-1}\end{pmatrix}\begin{pmatrix} g_a\\ g_b\end{pmatrix} = \frac12 \nabla f^T \P_1|_{\lambda = 1} \nabla g.
\]

Indeed, this evaluation produces a Poisson bracket for any dimension.
\begin{proposition}\label{red0} The reduction of (\ref{rightbr}) to $\K$ obtained through the formula
\[
\{f,g\}_0(\kb) = \{\F, \G\}'(K)
\]
is given explicitly by a $n\times n$ matrix whose $(i,j)$ entry is equal to: 
\begin{enumerate}
\item zero if $i+j > n+1$;
\item $\T^j - \T^{-i}$ if $i+j = n+1 $
\item $\T^j k^{i+j} - k^{i+j} \T^{-i}$ if $i+j < n+1$.
\end{enumerate} 
\end{proposition}
\begin{proof} This reduction can, in fact, be found explicitly. Assume
\begin{equation}\label{FQ}
\nabla_s \F(K) = \begin{pmatrix} Q_s &\frac{\partial f}{\partial \kb_s}\\ \q_s^T & -\tr(Q_s)\end{pmatrix}
\end{equation}
and assume $Q_s = (q^s_{ij})$ and $\q_s = (q_i^s)$. Let us call $\f = \frac{\partial f}{\partial \kb_s}:= (f^{i})$ and $\varpi = -\tr(Q_s)$. As before, we will drop the subindex $s$ to avoid cluttering. With this notation, and noticing that 
\[
K = \begin{pmatrix} I_n& 0\\ \kb^T& 1\end{pmatrix}\Lambda
\]
where
\[
\Lambda = \begin{pmatrix} 0&0&\dots &0&(-1)^n\\ 1&0&\dots&0&0\\ 0&1&\dots&0&0\\ \vdots&\ddots&\ddots&\ddots&\vdots\\ 0&\dots&0&1&0\end{pmatrix}
\]
We can write
\begin{equation}\label{FprimeQ}
\nabla'\F = \Lambda^{-1}\begin{pmatrix} Q+\f \kb^T & \f\\ \q^T-\kb^TQ + (\varpi-\kb^T \f)\kb^T&\varpi-\kb^T\f\end{pmatrix} \Lambda.
\end{equation}
Next we notice that conjugation by $\Lambda$ shifts rows once up and columns once to the left (with the first ones moving to last position, after multiplication by $(-1)^n$). That means the reduction of (\ref{rightbr}) is explicitly given by
\begin{equation}\label{formred0}
\{f, g\}_0(\kb) = \frac{(-1)^n}2\langle\T^{-1}\f, \begin{pmatrix}q^g_{12}\\\vdots\\ q^g_{1n} \\ g^1\end{pmatrix} 
+ g^1\begin{pmatrix}k^{2}\\\vdots\\ k^{n}\\0\end{pmatrix}\rangle - \frac{(-1)^n}2\langle\T^{-1}{\bf g}, 
\begin{pmatrix}q^f_{12}\\\vdots\\ q^f_{1n} \\ f^1\end{pmatrix} + f^1\begin{pmatrix}k^{2}\\\vdots\\ k^{n}\\0\end{pmatrix}
\rangle
\end{equation}
We will know the bracket explicitly once we find $q_{1j}$, $j = 2, \dots n$. These can be found without too much trouble.  Equation (\ref{streq2}) implies that the first $n$ rows of (\ref{streq2}) vanish. In our notation this can be written as
\begin{eqnarray*}
\T^{-1} q^f_{i j} &=& q^f_{i+1\ j+1} + f^{i+1} k^{j+1}, \hskip .1in j=1,\dots,n-1,\\
\T^{-1}q^f_{in} &=& f^{i+1}, \hskip .1in i=1,\dots,n-1,\\
\T^{-1}f_i &=& (-1)^n (q_{i+1\ 1}+f^{i+1} k^{i}).
\end{eqnarray*}
From here we have $q_{1n} = \T f^2$ and
\begin{equation}\label{q1j}
q_{1 j} = \T^{n-j+1}f^{n-j+2}+\sum_{p=2}^{n-j+1} \T^{p-1}(f^p k^{j-1+p})
\end{equation}
for any $j = 2,3,\dots, n-1$. Substituting  these values in (\ref{formred0}) we get the expression in the statement of the proposition. 
\end{proof}
\begin{theorem} The structure defined in proposition \ref{red0} is a Poisson structure.
\end{theorem}
\begin{proof}
We will prove that $\cH$ is a Hamiltonian operator using a theorem analogous to Theorem 7.8 in \cite{mr94g:58260} for difference operators
and following a process similar to the proof of Theorem \ref{th62}. Let ${\bf \theta}$ be a column vector with entry $\theta^{i}$. For the operator
$\cH$, we have $ \cH({\bf \theta})$ to be a column vector with entry
\begin{eqnarray*}
 {\cH({\bf \theta})}^{i}=\theta^{n+1-i}_{n+1-i}-\theta^{n+1-i}_{-i}
+\sum_{j=1}^{n-i} \left( (k^{i+j} \theta^{j})_j-k^{i+j} \theta_{-i}^{j} \right) .
\end{eqnarray*}
We now define the bi-vectors
\begin{eqnarray*}
&&\quad \Theta=\frac{1}{2} \int {\bf \theta} \wdg \cH ({\bf \theta})\\
&&=\frac{1}{2} \sum_{i=1}^n\!\! \int\!\!\! \left(\theta^{i}\wdg (\theta^{n+1-i}_{n+1-i}\!-\!\theta^{n+1-i}_{-i})
+\theta^{i} \wdg \sum_{j=1}^{n-i} \left( (k^{i+j} \theta^{j})_j \!-\!k^{i+j} \theta_{-i}^{j} \right) \right)\\
&&=\sum_{i=1}^n\!\! \int\!\!\! \theta^{i}\wdg \theta^{n+1-i}_{n+1-i}
+\sum_{i=1}^{n-1} \sum_{j=1}^{n-i} \int k^{i+j}  \theta^{i}_{-j} \wdg  \theta^{j} .
\end{eqnarray*}
We know that an anti-symmetric operator $\cH$ is Hamiltonian if and only if the tri-vector $\mbox{Pr}_{\cH({\bf \theta})}\Theta$ vanishes (see \cite{mr94g:58260}). We now show this is
the case for the given $\cH$. Indeed, we have
\begin{eqnarray*}
&&\quad \mbox{Pr}_{\cH({\bf \theta})}\Theta = \sum_{i=1}^{n-1} \sum_{j=1}^{n-i} \int {\cH({\bf \theta})}^{i+j}  
\theta^{i}_{-j} \wdg  \theta^{j}\\
&&= \sum_{i=1}^{n-1} \sum_{j=1}^{n-i} \int \left(\theta^{n+1-i-j}_{n+1-i-j}-\theta^{n+1-i-j}_{-i-j} \right) \wdg
\theta^{i}_{-j} \wdg  \theta^{j}\\
&&+ \sum_{i=1}^{n-1} \sum_{j=1}^{n-i} \sum_{l=1}^{n-i-j} \int \left( (k^{i+j+l} \theta^{l})_l
-k^{i+j+l} \theta_{-i-j}^{l} \right) \wdg \theta^{i}_{-j} \wdg  \theta^{j}\\
&&= \sum_{i=1}^{n-1} \sum_{j=1}^{n-i} \int \left(\theta^{n+1-i-j}_{n+1-i-j} \wdg \theta^{i}_{-j} \wdg  \theta^{j}
- \theta^{n+1-i-j}_{-i}  \wdg \theta^{i} \wdg  \theta^{j}_j \right)\\
&&+ \sum_{i=1}^{n-1} \sum_{j=1}^{n-i} \sum_{l=1}^{n-i-j}  k^{i+j+l} \int \left(\theta^{l}
 \wdg \theta^{i}_{-j-l} \wdg  \theta^{j}_{-l}-  \theta_{-i-j}^{l} 
 \wdg \theta^{i}_{-j} \wdg  \theta^{j}\right)\\
&&=0 
\end{eqnarray*}
by changing dummy variables. Thus we proved the statement.
\end{proof}
Using the general reduced operator we would like to write down an $n$-component integrable system. Notice that we have not proved in general that the two $n$-dimensional reduced brackets form a pencil. Nevertheless, we will be able to find an integrable system induced on $\kb$ by an invariant evolution of polygons in $\RP^n$.
First we notice the following fact.
\begin{proposition}
The two functionals $\ln k^{1}$ and 
$\frac{k^{2}_1}{k^{1} k^{1}_1}$ are in involution with respect to the Hamiltonian operator defined in proposition \ref{red0}; that is, $$\{ \ln k^{1},\   \frac{k^{2}_1}{k^{1} k^{1}_1}\}_0
=\int \left(\delta \ln k^{1}\right)^{T} \cH \delta \frac{k^{2}_1}{k^{1} k^{1}_1}=0.$$
\end{proposition}
\begin{proof}
This can be proved by direct calculations. First we have
\begin{eqnarray*}
&&\delta \ln k^{1}=(\frac{1}{k^{1}},0 ,\cdots, 0)^{T};\\
&&\delta \frac{k^{2}_1}{k^{1} k^{1}_1}=( -\frac{k^{2}_1}{(k^{1})^2 k^{1}_1}
-\frac{k^{2}}{(k^{1})^2 k^{1}_{-1}}, \frac{1}{k^{1} k^{1}_{-1}}, 0 ,\cdots, 0)^{T},
\end{eqnarray*}
which leads to 
\begin{eqnarray*}
 \left( \delta \frac{k^{2}_1}{k^{1} k^{1}_1}\right)^{T} \cH \delta \ln k^{1}
=-(\T-1) \left(\frac{k^{2}}{k^{1}k^{1}_{-1}}\right)^2
\end{eqnarray*}
and thus we obtain the result in the statement.
\end{proof}
From this proposition, we know that the corresponding Hamiltonian vector fields commute, that is,
$[\cH \delta \ln k^{1}, \ \cH \delta \frac{k^{2}_1}{k^{1} k^{1}_1}]=0$. 

Finally, we are going to show that the Hamiltonian system ${\bf k}_t=\cH \delta\ln k^{1}$ is integrable. 
Here we say a system is integrable if it possesses a hierarchy of infinitely many commuting symmetries.
We first write the system explicitly as
\begin{eqnarray}\label{neq}
\left\{ \begin{array}{l}k^{i}_t= \frac{k^{i+1}_1}{k^{1}_1}-\frac{k^{i+1}}{k_{-i}^{1}}, \quad i=1,2,\cdots, n-1\\ 
 k^{n}_t=\frac{1}{k^{1}_1}-\frac{1}{k^{1}_{-n}}\end{array}\right.
\end{eqnarray}
We then introduce the Miura transformation
\begin{eqnarray*}
 u^{1}=\frac{1}{k^{1} k^{1}_1 \cdots k^{1}_n}, \quad
u^{i}=\frac{k^{i}_{i-1}}{k^{1} k^{1}_1 \cdots k^{1}_{i-1}}, \quad i=2,3,\cdots, n.
\end{eqnarray*}
Under this transformation, equation (\ref{neq}) becomes ${\bf u}_t=-\tilde \cH \delta \frac{\ln u^{1}}{n+1}$, where 
$\tilde \cH=D_{\bf u} \cH D_{\bf u}^\star$ denote the transformed Hamiltonian operator. It is written as
\begin{eqnarray}\label{nequ}
\left\{ \begin{array}{l}u^{1}_t=-u^{1} (u^{2}_n-u^{2}_{-1})\\ 
u^{i}_t=u^{i+1}-u^{i+1}_{-1}-u^{i} (u^{2}_{i-1}-u^{2}_{-1}), \quad i=2, 3\cdots, n-1\\
u^{n}_t=u^{1}-u^{1}_{-1}-u^{n} (u^{2}_{n-1}-u^{2}_{-1})
\end{array}\right.
\end{eqnarray}
If we introduce the convention that $u^{j}=u^{j-n}$ if $j>n$, we can put the last two equations together as
\begin{eqnarray*}
u^{i}_t=u^{i+1}-u^{i+1}_{-1}-u^{i} (u^{2}_{i-1}-u^{2}_{-1}), \quad i=2, 3 \cdots, n, \ (u^{n+1}=u^{1}).
\end{eqnarray*}
Take $n=2$, $u^1=u$ and $u^2=-v$. Then equation (\ref{nequ}) becomes the Boussinesq lattice related to the lattice $W_3$-algebra (see (\ref{equv})). 

We observe that the vector ${\bf \tau}=(\tau^{1}, \cdots, \tau^{n})^{T}$ defined by
\begin{eqnarray*}
\begin{array}{l}\tau^{1}= s u^{1}_t -u^{1} \left((n+1) u^{2}_n+\sum_{l=0}^{n}u^{2}_{l}\right);\\ 
\tau^{i}= s u^{i}_t-u^{i} \left(i u^{2}_{i-1}+\sum_{l=0}^{i-1} u^{2}_{l}\right) +(i+1) u^{i+1}, \ \  i=2, 3\cdots, n,
\end{array}
\end{eqnarray*}
where $s$ is the independent discrete variable and we used the above convention $u^{n+1}=u^{1}$, 
is a master symmetry (see \cite{mr86c:58158} for definition)
of system (\ref{nequ}). Indeed, we have $[{\bf u}_t,\ \tau]\neq 0$ and  $[{\bf u}_t,\ [{\bf u}_t,\ \tau]]= 0$.
Thus, we can recursively generate the hierarchy of commuting symmetry flows of system (\ref{nequ}) by setting 
\begin{eqnarray}\label{symmetry}
 Q^0={\bf u}_t \quad \mbox{and} \quad Q^i=[\tau, \ Q^{i-1}].
\end{eqnarray}
Therefore, the system (\ref{nequ}) is integrable.
\begin{proposition} The master symmetry is a Hamiltonian vector with Hamiltonian function $-\frac{2 s+n+2}{2(n+1)} \ln u^{1}$. 
Moreover, all symmetries
generated by (\ref{symmetry}) are Hamiltonian vector fields. Their Hamiltonians are given by
\begin{eqnarray*}
 f^i=<\tau,\ \delta f^{i-1}> \quad \mbox{and} \quad f^0= -\frac{\ln u^{1}}{n+1}.
\end{eqnarray*}
\end{proposition}
\begin{proof}
By straightforward calculation,  we can check 
${\bf \tau}=-\frac{1}{n+1}\tilde \cH \delta \left(\frac{2 s+n+2}{2}\ln u^{1} \right)$. Therefore, 
the Hamiltonian operator $\tilde \cH $ is conserved along the vector field $\tau$ (see \cite{mr94j:58081}). In a word,
 the Lie derivative of $\tilde \cH$ along $\tau$ vanished, that is,  $L_{\tau} \tilde \cH=0$.
We know $Q^0=\tilde \cH f^0$.
Take the Lie derivative along the vector field $\tau$ on its both sides. It follows that 
\begin{eqnarray*}
Q^1=L_{\tau} Q^0 =L_{\tau} \left(\tilde \cH \delta f^0 \right)=\tilde \cH L_{\tau} \delta f^0=\tilde \cH <\tau,\ \delta f^0>
\end{eqnarray*}
We denote its Hamiltonian as $f^1$. So $Q^1$ is a Hamiltonian vector field with Hamiltonian $f^1$. By induction, we can
prove the statement.
\end{proof}
Indeed, we can compute the following Hamiltonians
\begin{eqnarray*}
 f^1&=&-\int\frac{\tau^{1}}{(n+1) u^{1}}= \frac{1}{n+1} \int \left(s (u^{2}_n-u^{2}_{-1}) + \left((n+1) u^{2}_n
+\sum_{l=0}^{n}u^{2}_{l}\right)\right)\\
&=&u^{2}\\
f^2&=&\int \tau^{2}=\int\left( s u^{2}_t-u^{2} \left(2 u^{2}_{1}+\sum_{l=0}^{1} u^{2}_{l}\right) 
+ 3 u^{3}\right)\\
&=&2  u^{3}-2 u^{2}  u^{2}_1 -(u^{2})^2\\
&&\qquad \cdots \cdots
\end{eqnarray*}
These are the integrals obtained in \cite{hika98} for the classical lattice $W_{n+1}$ algebra.
Thus equation (\ref{nequ}) is the integrable lattice related to the lattice $W_{n+1}$-algebra, which is not explicitly 
written down in \cite{hika98} for all $n$.

To finish the paper we will describe the projective realization of the $\RP^n$ completely integrable system. First a lemma that facilitates the description.

\begin{lemma}
$$ \HH(\delta \ln k^1) = - \P(\delta\ln k^1) $$
\end{lemma}
\begin{proof}
From the reduction process, if we want to find $\P(\frac{\partial f}{\partial \kb}) = \P((-1)^n\frac 1{k^1}e_1)$, we need to find
\begin{equation}\label{prestar}
\nabla \F - \T\nabla'\F = \begin{pmatrix} 0&0\\ \P(\f)^T & 0\end{pmatrix}
\end{equation}
where $\f = \frac {(-1)^n}{k^1}e_1$ and $\F$ is a proper extension of $f(\kb) = \ln k^1$. Using (\ref{FQ}) and (\ref{FprimeQ}) we can transform this equality into
\[
\Lambda\nabla \F-\T\begin{pmatrix}1&0\\ -\kb^T&1\end{pmatrix} \nabla\F\begin{pmatrix}1&0\\ \kb^T&1\end{pmatrix} \Lambda = \begin{pmatrix} (-1)^n \P(\f) & 0\\ 0&0\end{pmatrix}.
\]
Expression (\ref{prestar}) can be written in more detail as
\begin{equation}\label{star}
\begin{pmatrix} \mathsmaller{(-1)^n\q^T}&\mathsmaller{(-1)^n \varpi}\\ \mathsmaller{Q}&\mathsmaller{\f}\end{pmatrix} -\begin{pmatrix}\mathsmaller{\T\tilde Q+\T(\widetilde{f_1\kb^T})}&\mathsmaller{\T\f}&\mathsmaller{(-1)^n(\T Qe_1+\T (k^1\f))}\\\mathsmaller{ \T\tilde q^T-\T(\widetilde{\kb^TQ}) + \T((\varpi-k\cdot \f)\tilde\kb^T)}&\mathsmaller{\T(\varpi-\kb\cdot \f)} & \mathsmaller{(-1)^n(\T q_1-\T\kb^TQe_1+\T((\varpi-\kb\cdot\f)k_1))}\end{pmatrix}
\end{equation}
\[
= \begin{pmatrix} (-1)^n\P(\f) & 0 \\ 0&0\end{pmatrix}
\]
where the tilde indicates that the first column has been removed.  From here we get
\begin{equation}\label{Pf}
\P(\f) = \q - (-1)^n \begin{pmatrix} \T(\overline{Q^Te_1}) + \T(f_1 \bar \kb)\\ \T f_1\end{pmatrix},
\end{equation}
where the bar indicates that the first row has been removed. Recall from (\ref{q1j}) that if $Q = (q_{ij})$ and $\f= \frac{(-1)^n}{k^1} e_1$, then $q_{1j} = 0$ for $j\ne 1$. This means $\overline{Q^Te_1} = 0$. With straightforward calculations we can also get  that $q_{11} = -\frac1{n+1}$. Using 
\begin{eqnarray*}
q_1 - \kb^TQe_1 +(\varpi-1)k_1 & = & 0\\ (-1)^n(\T\overline{Q e_1} + \T\overline{k_1\f}) & = & \frac{(-1)^n}{k^1}e_1,
\end{eqnarray*}
we obtain 
\[
\varpi -1 = -\frac1{n+1}, \hskip1ex Qe_1 = -\frac 1{n+1} e_1 + \T^{-1}\frac1{k^1} e_2, \hskip 1exq_1 = k^2\T^{-1} \frac1{k^1}.
\]

A recursive use of the lower left block in (\ref{star}) produces the values of $Q e_s$. They are given by \[
Qe_s = -\frac1{n+1} e_s + \T^{-1}\frac1{k^1}e_{s+1}, \hskip 1ex r=1,\dots, n-1, \hskip 2ex Qe_n = -\frac1{n+1} e_n.
\]
 The last row of that block gives the value of $\q$. They are given by
\[
q_i = k^{i+1}\T^{-i}\frac1{k^1}.
\]
Substituting the values in (\ref{Pf}) proves the lemma.
\end{proof}
Assume $V_s$ are the original lifts of our projective polygon and $\hat k^i$ are the invariants given by the relation $V_{s+n+1} = \hat k_s^n V_{s+n} + \dots + \hat k_s^1 V_{s+1} +(-1)^n V_s$.
\begin{theorem} The projectivization of the evolution
\[
(V_s)_t = \frac{-1}{\hat k_s^1}\left(V_{s+n} - \hat k_{s-1}^n V_{s+n-1}-\dots -\hat k_{s-1}^2V_{s+1}\right) + v^0_s V_s
\]
induces  the completely integrable system (\ref{neq}) on the gauged invariants $\kb$.
\end{theorem}
\begin{proof}
Since the $\P$-Hamiltonian for the system is $f(\kb) = \ln k^1$, we know that the lift of the projective realization is given by
\[
(V_s)_t = \rho_s^{-1} \v_s
\]
where $\v_s = \frac{\partial f}{\partial \kb} = \frac {(-1)^n}{k^1} e_1$. Since $\rho_s = (W_{s+n}, \dots, W_s)$ as in (\ref{lift}), we have 
\[
(V_s)_t = \frac{(-1)^n}{k_s^1} W_{s+n} + v^0_sV_s.
\]
Checking the gauge carefully we see that $k^1 = (-1)^{n-1} \hat k^1$ and $W_{n+r} = V_{n+r} - \hat k_{s-1}^n V_{n+r-1} - \dots - \hat k_{s-1}^2 V_{s+1}$, which concludes the proof. \end{proof}

It is only natural to conjecture that both of our reductions form a Hamiltonian pencil associated to integrable discretizations of $W_n$ algebras. 
The fact that one of them is not a Poisson bracket originally (and one cannot make it so in general using $\g_0$ perturbations) seriously complicates the proof of this conjecture.

\end{document}